\documentclass{article} 

\usepackage{geometry}
\usepackage{authblk}
\usepackage{amsmath,amsthm}
\usepackage[english]{babel}
\usepackage[utf8x]{inputenc}
\usepackage[T1]{fontenc}
\usepackage{enumerate}
\usepackage[ruled,vlined,linesnumbered]{algorithm2e}
\usepackage{nicefrac}
\usepackage{thm-restate}
\usepackage{multicol}
\usepackage{amsmath,amssymb,mathtools}
\usepackage{graphicx}
\usepackage[colorinlistoftodos]{todonotes}
\newcommand{\R}{\mathbb{R}}
\newcommand{\Q}{\mathbb{Q}}
\newcommand{\rk}{\mathrm{rk}}
\newcommand{\supp}{\mathrm{supp}}
\newcommand{\bp}{\bar p}
\newcommand{\MBB}{\mathrm{MBB}}
\newif\ifnotes\notestrue
\usepackage{tikz,pgfplots}
\usetikzlibrary{positioning,arrows,patterns,decorations.pathreplacing,calc,decorations.pathmorphing,positioning,calc,shapes,snakes,intersections,arrows.meta}

\newcommand{\notename}[2]{{}}

\newcommand{\jnote}[1]{}
\newcommand{\lnote}[1]{}

\newcommand{\Z}{Z_{+}}
\newcommand{\ot}{\leftarrow}

\SetCommentSty{mycommfont}
\renewenvironment{proof}{ \noindent {\it Proof.}$\;$\rm}{\hfill $\Box$\medskip}

\newtheorem{theorem}{Theorem}[section]
\newtheorem{lemma}[theorem]{Lemma}
\newtheorem{corollary}[theorem]{Corollary}

\newtheorem{claim}[theorem]{Claim}
\theoremstyle{definition}
\newtheorem{definition}[theorem]{Definition}

\newtheorem{remark}[theorem]{Remark}

\newcommand{\examplebox}[1]{%
\begin{center}\fbox{
\begin{minipage}{0.8\textwidth}
#1
\end{minipage}
}\end{center}}

\begin{document}

\title{A Strongly Polynomial Algorithm for Linear Exchange Markets\thanks{Financial support from the Division of Computing and Communication Foundations, National Science Foundation (NSF) [Grant 1755619] and the H2020 European Research Council (ERC) [Grant ScaleOpt--757481] is gratefully acknowledged.}}

\author[1]{Jugal Garg}
\author[2]{L\'{a}szl\'{o} A. V\'{e}gh}

\affil[1]{Department of Industrial and Enterprise Systems Engineering, University of Illinois at Urbana-Champaign}
\affil[2]{Department of Mathematics, London School of Economics}

\date{\tt{jugal@illinois.edu,l.vegh@lse.ac.uk}}

\maketitle

\begin{abstract}%
We present a strongly polynomial algorithm for computing an equilibrium in Arrow-Debreu exchange markets with linear utilities. 
Our algorithm is based on a variant of the weakly-polynomial 
 Duan--Mehlhorn (DM) algorithm. We use the DM algorithm as a
 subroutine to identify  revealed edges, i.e. pairs of agents and goods that must correspond to best
 bang-per-buck transactions in every equilibrium solution.
Every time a new revealed edge is found, we use another subroutine
that decides if there is an optimal solution using the current set of
revealed edges, or if none exists, finds the solution that
approximately minimizes the violation of the demand and supply
constraints.
This task can be reduced to solving a linear program (LP). Even though we are unable to solve this LP in strongly polynomial time, we show that it can be approximated by a simpler LP with two variables per inequality that is solvable in strongly polynomial time. \end{abstract}

\section{Introduction}
Market equilibrium is a central concept in economics to analyze and predict the outcome of agents' interactions in markets. It has found applications in a variety of domains, even in \emph{non-market settings} which may not involve any exchange of money but only require the remarkable fairness and efficiency properties of equilibria, see e.g.~\cite{Varian74}. 

Exchange is one of the most fundamental market models, introduced by Walras~\cite{walras}. 
An exchange market is like a barter system, where a set of agents arrive at a market with an initial
endowment of divisible goods and have a 
utility function over
allocations of goods. 
Agents can use their revenue from selling their initial endowment to purchase
their preferred bundle of goods. In a market equilibrium, the prices
are such that each agent can spend their entire revenue on a bundle of
goods that maximizes her utility at the given prices, and all goods
are fully sold.

The celebrated result of Arrow and Debreu \cite{arrow} shows the 
existence of an equilibrium for a broad class of utility functions.
Computational aspects have been already addressed since the
19th century (see e.g., Brainard and Scarf~\cite{BrainardS00}), and polynomial time algorithms have been
 investigated 
over the last twenty years; see the survey by Codenotti et al.~\cite{codenotti04} for early
work, and the references in Garg et al.~\cite{GargMVY17} for more recent developments.

In this paper we study the case where all utility functions  are \emph{linear}.
Linear market models have been extensively studied since 1950s; see~\cite{DevanurGV16} for an overview of earlier work.
These models are also appealing from a combinatorial optimization
perspective due to their connection to classical network flow models and
their rich combinatorial structure. A well-studied special case of the
exchange market is the {\em Fisher market} setting, where every buyer
arrives with a fixed budget instead of an endowment of goods. Using
network flow techniques, Devanur et al.  \cite{DevanurPSV08} gave a polynomial-time combinatorial algorithm
that was followed by a series of further such algorithms (for example Vazirani \cite{Vazirani10}, Goel and Vazirani~\cite{GoelV11}), including
strongly polynomial ones: Orlin~\cite{orlin}, V\'egh~\cite{Vegh16}. A combinatorial algorithm for the general
exchange market was developed later by  Duan and Mehlhorn
\cite{DuanM15}, and no strongly polynomial
algorithm has been known thus far.

\paragraph{Strongly polynomial algorithms and rational convex programs.} Assume that a problem is given by an input of $N$ rational numbers
given in binary description. An algorithm for such a problem is
\emph{strongly polynomial} (see Section 1.3 in \cite{glsbook}), if 
it only uses elementary arithmetic operations (additions, 
  multiplications,  divisions, and comparisons), and the total number of such operations is bounded by
  poly($N$). Further, the algorithm is required to run in \emph{polynomial space}: that is, the size
  of the numbers occurring throughout the algorithm remain
  polynomially bounded
  in the size of the input. Here, the size of a rational number $p/q$ with integers
$p$ and $q$ is defined as
$\lceil\log_2(|p|+1)\rceil+\lceil\log_2(|q|+1)\rceil$.

It is a major open question to find a strongly polynomial algorithm for
linear programming. Such algorithms are known for special classes of 
 linear
optimization problems. We do not present a comprehensive overview here
but only highlight some examples: systems of linear equations with at
most two nonzero entries per inequality \cite{Adler91,CohenM94,Megiddo83}; minimum cost circulations
e.g. \cite{Goldberg89,orlin93,Tardos85}; LPs with
bounded entries in the constraint matrix: \cite{Tardos86,Vavasis96};
generalized flow maximization \cite{Vegh17,OV17}, and variants of Markov
Decision Processes: \cite{Ye2005,Ye2011}.

For nonlinear convex optimization,  only sporadic results are known.
 The relevance of certain market equilibrium problems
in this context is that they can be described by \emph{rational convex
  programs}, where a rational optimal solution exists with 
encoding size bounded in the input size (see
Vazirani~\cite{vazirani12}). This property gives hope for
finding strongly polynomial algorithms.

The linear Fisher market equilibrium can be captured by two different
convex programs, one by  Eisenberg and Gale  \cite{EG}, and one by
Shmyrev  \cite{Shmyrev09}. These are special cases of natural convex extensions of
classical network flow models by V\'egh~\cite{Vegh14,Vegh16}. In particular, the
second model is a network flow problem with a separable convex cost
function; \cite{Vegh16} provides a strongly polynomial algorithm for
the linear Fisher market using this general perspective.

No such description is known for the exchange market model. A rational convex program was given in Devanur et al.~\cite{DevanurGV16},
but the objective is not separable and hence the result in
\cite{Vegh16} cannot be applied. Previous convex programs (Nenakov and Primak~\cite{NenakovP83}, Cornet \cite{cornet89}, Jain\cite{Jain07}) included
nonlinear constraints and did not appear amenable for a combinatorial
approach (see \cite{DevanurGV16} for an overview).

\paragraph{Model.}
Let $A$ be the set of $n$ agents and $G$ a set of goods. Without loss of generality, we can
assume that there is one unit in total of each divisible good, and that
there is a one-to-one correspondence between agents and goods: agent
$i$ brings the entire unit of the $i$-th good, $g_i$ to the market. If agent $i$
buys $x_{ij}$ units of good $g_j$, her utility is $\sum_j
u_{ij}x_{ij}$, where $u_{ij}$ is the utility of agent $i$ for a unit amount of good $g_j$. Given prices
$p=(p_i)_{i\in A}$, the bundle that maximizes the utility of agent 
$i$ is any choice of \emph{maximum bang-per-buck goods}, that is,
goods that maximize the ratio $u_{ij}/p_j$. The prices $p$ and allocations
$(x_{ij})_{i,j\in A}$ form a market equilibrium, if {\em (i)}
$\sum_{i\in A}
x_{ij}=1$ for all $j$, that is, every good is fully sold; {\em (ii)}
$p_i=\sum_{j\in A} p_jx_{ij}$ for all $i$, that is, every agent spends
her entire revenue; and {\em (iii)} $x_{ij}>0$ implies that
$u_{ij}/p_j=\max_k u_{ik}/p_k$, that is, all purchases maximize bang-per-buck.

\paragraph{Algorithms for the linear exchange market.}
A finite time algorithm based on Lemke's scheme  \cite{Lemke} was obtained by
 Eaves \cite{Eaves76}. A necessary and sufficient condition for 
the existence of equilibrium was described by Gale  \cite{Gale76}.
The first polynomial-time algorithms for the problem were given by
 Jain  \cite{Jain07} using the ellipsoid method and by Ye  \cite{Ye08}
using an interior point method.  
The combinatorial algorithm of  Duan and Mehlhorn  \cite{DuanM15} builds on the algorithm \cite{DevanurPSV08} for
linear Fisher markets. An improved variant was given by Duan et al.~\cite{DuanGM16} with the current best weakly polynomial running time
$O(n^7 \log^3(nU))$, assuming all $u_{ij}$ values are integers between
$0$ and $U$. A main measure of progress in these
algorithms is the increase
in prices. However, the upper bound on the prices
depends on the $u_{ij}$ values; therefore, such an
analysis can only provide a weakly polynomial bound. The existence of
a strongly polynomial algorithm is described by Duan and Mehlhorn~\cite{DuanM15} as a
major open question. There is a number of simple algorithms for
computing an approximate equilibrium \cite{DevanurV03,GargK06,GhiyasvandO12,JainMS03}, but these do not give
rise to polynomial-time exact algorithms. 

\subsection{Overview of the algorithm} 
We provide a strongly polynomial algorithm for computing an exact equilibrium in linear exchange markets with running time $O(n^{10}\log^2 n)$.  
Let $F^*$ denote the set of edges (pairs of agents and goods) that
correspond to a best bang-per-buck transaction in every market equilibrium. 
In the algorithm, we maintain a set $F\subseteq F^*$ called
\emph{revealed edges}, and the main progress is adding a new edge in
strongly polynomial time. At a high level, this approach resembles
that of V\'egh \cite{Vegh16}, which
extends Orlin's \cite{orlin93} approach for minimum-cost circulations.

In a  \emph{money allocation} $(p,f)$, $(p_i)_{i\in A}$ is a set of prices
and $(f_{ij})_{i,j\in A}$ represent the amount of money paid for good
$g_j$ by agent $i$; $f_{ij}$ may only be positive for maximum
bang-per-buck pairs. In the algorithm we work with a money allocation
where goods may not be fully sold and agents
may have leftover money; we let $s(p,f)\in\R^G$ denote surplus vector.
Thus, $\|s(p,f)\|_{_1}$ is the total surplus
money left, and $\|s(p,f)\|_{_\infty}$ is the maximum surplus of any
good. We show that if $f_{ij}>\|s(p,f)\|_{_1}$ in a money allocation, then
$(i,g_j)\in F^*$ (Lemma~\ref{lem:F-extend}). This is analogous to  the method of identifying
\emph{abundant arcs} for minimum cost flows by Orlin  \cite{orlin93}.

We use a variant of the Duan-Mehlhorn (DM) algorithm to identify
revealed edges. We show that the potential  $\phi(p,f)=\|s(p,f)\|_{_\infty}/(\prod_{j}
p_j)^{1/n}$  decreases geometrically in the algorithm; from this fact,
it is not difficult to see that an edge with $f_{ij}>\|s(p,f)\|_{_1}$ appears
in a strongly polynomial number of iterations, yielding the first
revealed edge. We need to modify the DM algorithm \cite{DuanM15} so that, among other
reasons, the potential decreases geometrically when we run the algorithm starting with 
any arbitrary price vector $p$; see Remark~\ref{rem:difference} for all the differences.

To identify subsequent revealed edges, we need a more flexible
framework and a  second subroutine. We
work with the more general concept of $F$-allocations, where the money
amount $f_{ij}$ could be negative if $(i,g_j)\in F$. This is a viable
relaxation since an $F$-equilibrium (namely, a market equilibrium with
possibly negative allocations in $F$) can be efficiently turned into a
proper market equilibrium, provided that $F\subseteq F^*$ (Lemma~\ref{lem:F-eq}). 
Given a set $F$ of revealed edges, our \emph{Price Boost} subroutine
finds an approximately optimal $F$-allocation using only edges in
$F$. Namely, the
subroutine finds an $F$-equilibrium if there exists one; otherwise, it
finds an $F$-allocation that is zero outside $F$, and subject to this, it
approximately minimizes $\phi(p,f)$. This will provide the initial
prices for  the next iteration of the DM
subroutine. Since DM decreases $\phi(p,f)$ geometrically, after a
strongly polynomial number iterations it will need to send a
substantial amount of flow on an edge
outside $F$, providing the next revealed edge.

Let us now discuss the \emph{Price Boost} subroutine. The analogous
subproblem for Fisher markets in \cite{Vegh16} reduces to a simple variant of the
Floyd-Warshall algorithm. For exchange markets, we show that
optimizing $\phi(p,f)$ can be captured by a linear program. A strongly
polynomial LP algorithm would therefore immediately provide the
desired subroutine. 
Alas, this LP is not contained in any special class of LP where
strongly polynomial algorithms are currently known.

We will exploit the following special structure of the LP. First, we can
eliminate the $f_{ij}$ variables and only work with price variables.
We obtain a form where the objective is to maximize the sum of all variables over a feasible set of the
form $P\cap P'$. The first polyhedron $P$ is a \emph{monotone two variable per inequality (M2VPI)} system: there is one positive and one
nonnegative variable per inequality. 
The constraint matrix defining the  second polyhedron $P'$ is what we call
a \emph{$\Z$-matrix}: all
off-diagonal elements are nonpositive but all column sums are
nonnegative. This corresponds to a submatrix of the transposed of a
weighted Laplacian matrix.
A pointwise maximal solution to an M2VPI system can be found in strongly polynomial time using classical results by
\cite{Adler91,CohenM94,HochbaumN94,Megiddo83}.
To deal with the constraints defining $P'$, we approximate our LP by a second LP that can
be solved in strongly polynomial time.

More precisely, we replace the second polyhedron $P'$ by $Q$ such
that $\frac{1}{n^2} Q\subseteq P' \subseteq Q$, and that $Q$ is also a system
with one positive and one nonnegative variable per inequality. Thus,
an M2VPI algorithm is applicable to
maximize the sum of the variables over $P\cap Q$ in strongly
polynomial time. For an optimal
solution $\bp$, the vector $\bp/n^2$ is feasible to the original LP
and the objective value is within a factor $n^2$ of the optimum.
 For the purposes of identifying a new revealed edge in the algorithm,
 such an approximation of the optimal $\phi(p,f)$ value already suffices.

The construction of the approximating polyhedron $Q$ is obtained via a general
method applicable for systems given by $\Z$-matrices. We show that for
such systems, Gaussian elimination can be used to generate valid
constraints with at most two nonzero variables per row. Moreover, we
show that the intersection of all relevant such constraints provides a
good approximation of the original polyhedron.

\medskip

The rest of the paper is structured as
follows. Section~\ref{sec:prelim} introduces basic definitions and
notation. Section~\ref{sec:overall} describes the overall algorithm by
introducing the notion of $F$-allocations, the main potential, and the two necessary
subroutines. Section~\ref{sec:F} proves the lemmas necessary for identifying
revealed edges. Section~\ref{sec:dm} presents the first of these two
subroutines, a variant of the Duan-Mehlhorn
algorithm. Section~\ref{sec:boost} shows how the second subroutine,
 Price Boost, can be reduced to solving an
 LP. Section~\ref{sec:approximate} exhibits the polyhedral
   approximation result for $\Z$-matrices. Section~\ref{sec:conclude} concludes with some open questions.

\section{Preliminaries}\label{sec:prelim}
For a positive integer $t$, we let $[t]:=\{1,2,\ldots,t\}$, and for
$k<t$, we let $[k,t]:=\{k,k+1,\ldots,t\}$. For a
vector $a\in \R^n$, we let
\[
\|a\|_{_1}:=\sum_{i=1}^n |a_i|,\quad \|a\|_{_2}:=\sqrt{\sum_{i=1}^n
  a_i^2},\quad \mbox{and}\quad \|a\|_{_\infty} := \max_{i\in [n]}|a_i|
\]
denote the $\ell_1$, $\ell_2$, and $\ell_\infty$-norms,
respectively.
Further, we let $\supp(a)\subseteq [n]$ denote the
support of the vector $a\in \R^n$, that is, $\supp(a):=\{i\in [n]\ |\ a_i\neq
0\}$. For a subset $S\subseteq [n]$, we let $a(S):=\sum_{i\in S} a_i$.

\paragraph{The linear exchange market}
We let $A:=[n]$ denote the set of agents, 
$G:=\{g_1,g_2,\ldots,g_n\}$ denote the set of goods, and 
$u_{ij}\ge 0$ denote the utility of agent $i$ for a unit amount of good $g_j$. 
Let $E\subseteq A\times G$ denote the set of pairs $(i,g_j)$ such that
$u_{ij}>0$; let $m:=|E|$. We will assume that for each $i\in A$ there exists a 
$g_j\in G$ such that $u_{ij} >0$, and for each $g_j\in G$ there exists an $i\in A$
such that $u_{ij} >0$. Hence, $m\ge n$.

The goods are
divisible and there is one unit of each in total. Agent $i$ arrives
to the market with her initial endowment comprising exactly the unit of
good $g_i$. As mentioned in the introduction, 
the general case with an
arbitrary set of goods and arbitrary initial endowments can be easily
reduced to this setting; see~\cite{Jain07,DevanurGV16}. 

Let
$p=(p_j)_{g_j\in G}$ denote the prices where $p_j$ is the price of good $g_j$.
Given prices $p$, we let 
\[
\alpha_i:=\max_{g_k\in G} \frac{u_{ik}}{p_k}, \quad \forall i\in A.
\]
For each agent $i\in A$, the goods satisfying equality here are called
\emph{maximum bang-per-buck} (MBB) goods; for such a good $g_j$,
$(i,g_j)$ is called an MBB-edge.
We let $\MBB(p)\subseteq E$ denote the set of MBB-edges at prices $p$.

\begin{definition} Let $f=(f_{ij})_{i\in A, g_j\in G}$ denote the money flow where
$f_{ij}$ is the money spent by agent $i$ on good $g_j$.
 We say that $(p,f)$ is a \emph{money allocation} if 
\begin{enumerate}[(i)]
\item $p>0$, and $f\ge 0$;
\item $\supp(f)\subseteq \MBB(p)$;
\item  $\sum_{g_j\in
    G} f_{ij}\le p_i$ for every agent $i\in A$;
\item   $\sum_{i\in A} f_{ij}\le p_j$ for every good $g_j\in
  G$.
\end{enumerate}
\end{definition}
For the money allocation $(p, f)$, the \emph{surplus} of agent $i$ is
defined as
\[c_i(p, f) := p_i -
\sum_{g_j\in G} f_{ij},\]
 and the surplus of good $g_j\in G$ is defined as
\[ s_j(p, f) := p_j - \sum_{i\in A} f_{ij}.
\]
Parts (iii) and (iv) in the definition of the money allocation require
that the surplus of all agents and goods are nonnegative.
We let $c(p,f):=(c_i(p,f))_{i\in A} $ and $s(p,f):=(s_j(p,f))_{g_j\in
  G}$ denote the surplus vectors of the
agents and the goods, respectively.
Clearly, $\|s(p,f)\|_{_1}=\|c(p,f)\|_{_1}$, and $\|s(p,f)\|_{_\infty}\le \|s(p,f)\|_{_1}\le n\|s(p,f)\|_{_\infty}$.

\begin{definition}
A money allocation $(p,f)$ is called a \emph{market equilibrium} if $\|s(p,f)\|_{_1}=0$.
\end{definition}

Note that if $(p,f)$ is a market equilibrium and $\lambda>0$ is any positive number, then
$(\lambda p,f)$ is also a market equilibrium. We will refer to this as the \emph{price scaling property of market equilibria.}
\paragraph{Existence of an equilibrium}
A market equilibrium may not exist for certain inputs.
  For example, consider a market with 3 agents and 3 goods, where
  agent $i$ brings one unit of good $g_i$ for $i=1,2,3$. Let
  $E=\{(1,g_1), (1,g_2), (2,g_3), (3,g_3)\}$. Clearly, in this market, at any prices, Agent 3 will consume
  the entire $g_3$. If $p_2 > 0$, then Agent 2 will demand a non-zero
  amount of $g_3$, and therefore no market equilibrium exists.

A necessary and sufficient condition for the existence of an
equilibrium can be given as follows. Let us define the directed graph
$(A,\bar E)$, where $(i,j)\in \bar E$ if and only if $u_{ij}>0$ (that
is, if $(i,g_j)\in E)$. 
\begin{lemma}[\cite{Gale76,DevanurGV16}]
There exists a market equilibrium if and only if for every strongly
connected component $S \subseteq  A$ of the digraph $( A, \bar E)$, if
$|S| = 1$, then there is a loop incident to the node in $S$.
\end{lemma}

This condition can be easily checked in strongly polynomial time. 
Further, it is easy to see that if the above condition holds, then
finding an equilibrium in an arbitrary input can be reduced to finding
an equilibrium in an input where the digraph $(A,\bar E)$ is strongly
connected. Thus, we will assume the following throughout the paper:

\begin{equation}\label{eq:strongly}\tag{$\star$}
\mbox{The graph $(A,\bar E)$ is strongly connected.}
\end{equation}

\section{The overall algorithm}\label{sec:overall}
In this section, we describe the overall algorithm. We formulate the
statements that are needed to prove our main theorem:
\begin{theorem}\label{thm:main}
There exists a strongly polynomial algorithm that computes a market
equilibrium in linear exchange markets in time $O(n^{10}\log^2 n) $.
\end{theorem}
 
We start by introducing the  concepts of revealed edges, $F$-allocations, and balanced $F$-flows.
\subsection{Revealed edges}
Throughout the paper, we let $F^*\subseteq E$ denote the set of edges 
$(i,g_j)\in E$ that are MBB edges for \emph{every} market equilibrium
$(p,f)$.
Throughout the algorithm, we maintain a subset of edges $F\subseteq F^*$
 called the \emph{revealed edge set}. This is
initialized as $F=\emptyset$.

\begin{definition}
 For an edge set $F\subseteq E$,  $(p,f)$ is an \emph{$F$-money allocation}, or \emph{$F$-allocation} in short, if
\begin{enumerate}[(i)]
\item $p>0$, and $f_{ij}\ge 0$ for $(i,g_j)\in E\setminus F$;
\item $\supp(f)\cup F\subseteq \MBB(p)$;
\item $\sum_{g_j\in
    G} f_{ij}\le p_i$ for every agent $i\in A$;
\item 
  $\sum_{i\in A} f_{ij}\le p_j$ for every good $g_j\in G$.
\end{enumerate}
An $F$-allocation is called an \emph{$F$-equilibrium} if $\|s(p,f)\|_{_1}=0$.
\end{definition}
Note
that $f_{ij}$ could be negative for
$(i,g_j)\in F$.  A $\emptyset$-allocation simply corresponds to a
money allocation.

The main progress step in the algorithm is adding new edges to
$F$. This is enabled by the following lemma; the proof is
deferred to Section~\ref{sec:F}.

\begin{lemma}\label{lem:F-extend}
Let $F\subseteq F^*$, and let $(p,f)$ be an
$F$-allocation. If $f_{k\ell} > \|s(p,f)\|_{_1}$ for an edge $(k,g_\ell)\in E$, then $(k,g_\ell)\in F^*$. 
\end{lemma}

Our algorithm obtains an $F$-equilibrium. 
Whereas an $F$-equilibrium is not necessarily a market equilibrium, the following holds true:

\begin{lemma}\label{lem:F-eq}
Let $F\subseteq F^*$, and assume we are given an $F$-equilibrium
$(p,f)$. Then a market equilibrium $(p,f')$ can be obtained in $O(nm)$
time. 
\end{lemma}

We let \textsc{Final-Flow}$(p)$ denote the algorithm as in the
Lemma. This is a maximum flow computation in an auxiliary network, as
described in the proof in Section~\ref{sec:F}.

\subsection{Balanced flows}\label{sec:balanced}
Balanced flows play a key role in the Duan-Mehlhorn algorithm 
\cite{DuanM15}, as well as in previous algorithms for 
Fisher market models, e.g., \cite{DevanurPSV08,GoelV11,Vazirani10}. We now
introduce the natural extension for $F$-allocations.
\begin{definition}
Given an edge set $F\subseteq E$ and prices $p$, we say that $(p,f)$ is a \emph{balanced
  $F$-flow}, if $(p,f)$ is an $F$-allocation that minimizes $\|s(p,f)\|_{_1}$,
and subject to that, it minimizes $\|c(p,f)\|_{_2}$. 
\end{definition}

\begin{lemma}[\cite{Kamiyama19,DarwishM16}]\label{lem:balanced}
Given $F\subseteq E$ and prices $p$ such that $F\subseteq \text{MBB}(p)$, a
balanced $F$-flow can be computed in $O(nm\log{(n^2/m)})$ time.
\end{lemma}
Whereas the original form does not include a revealed set $F$, the extension to $F$-flows is immediate.
We let \textsc{Balanced}$(F,p)$ denote the subroutine guaranteed by
the Lemma.

\subsection{The algorithm}
\begin{algorithm}[t]
\caption{Arrow-Debreu Equilibrium}\label{alg:oa}
\DontPrintSemicolon
\SetKwInOut{Input}{Input}\SetKwInOut{Output}{Output}
\Input{Set $A$ of agents, set $G$ of goods, and utilities $(u_{ij})_{i\in A, g_j\in G}$}
\Output{Market equilibrium $(p,f)$} 
$F\ot \emptyset$;
\;
\Repeat{$\|s(p,f)\|_{_1} = 0$}{
$(\hat p,\hat f) \ot$ \textsc{Boost}($F$)\tcp*{Theorem~\ref{thm:boost}}
$(p,f) \ot$ \textsc{DM}($F,\hat p$) \tcp*{Theorem~\ref{thm:dm}}
$F\ot F\cup \{(i,g_j)\ |\ f_{ij} > \|s(p,f)\|_{_1}\}$ \tcp*{Lemma~\ref{lem:F-extend}}
}
$f\ot $\textsc{Final-Flow}($p$) \tcp*{Lemma~\ref{lem:F-eq}}
\Return{$(p,f)$}\;
\end{algorithm}

The overall algorithm is presented in Algorithm~\ref{alg:oa}.
The main progress is gradually expanding a revealed edge set
$F\subseteq F^*$, initialized as $F=\emptyset$.
Every cycle of the algorithm performs the subroutines 
\textsc{Boost}$(F)$ and \textsc{DM}$(F,\hat p)$, and at least one new edge
is added to $F$ at every such cycle. 
Once an $F$-equilibrium is
obtained for the current $F$, we use the subroutine
\textsc{Final-Flow}$(p)$ as in Lemma~\ref{lem:F-eq} to compute a market
equilibrium.

We now introduce the key potential measures used in analysis. For an $F$-allocation 
$(p,f)$, we define 
\[\phi(p, f) := \displaystyle\frac{\|s(p,f)\|_{_\infty}}{(\prod_{j=1}^n
  p_j)^{1/n}}.\]
Note that this is invariant under scaling, i.e. $\phi(p,f)=\phi(\lambda
p,\lambda f)$ for any $\lambda>0$. Further, $(p,f)$ is an
$F$-equilibrium if and only if $\phi(p,f)=0$.
For a given $F\subseteq F^*$, we define
\begin{equation}\label{psi-def}
{\small
\Psi(F):=\min\{\phi(p,f):\, (p,f)\mbox{ is an $F$-allocation,
}\supp(f)\subseteq F\}.
}
\end{equation}
\begin{theorem}\label{thm:boost}
There exists a strongly polynomial time algorithm that for any input
$F\subseteq E$, returns in time $O(n^4\log n)$ an $F$-allocation $(\hat p,\hat f)$ with
$\supp(\hat f)\subseteq F$ such that $\Psi(F)\le \phi(\hat p,\hat f)\le (n-1)^2 \Psi(F)$.
\end{theorem}
The algorithm in the theorem will be denoted as \textsc{Boost}($F$),
and is described in Section~\ref{sec:boost}.
In particular, if $\Psi(F)=0$, then \textsc{Boost}($F$) returns an
$F$-equilibrium.

The second main subroutine \textsc{DM}($F,\hat p$), is a variant of
 the algorithm  by \cite{DuanM15}, described in
Section~\ref{sec:dm}.
 As the input, it uses the prices $\hat p$ obtained in
the $F$-allocation $(\hat p,\hat f)$ returned by \textsc{Boost}($F$),
and outputs an $F$-allocation $(p,f)$ such that either $\|s(p,f)\|_{_1}=0$,
that is, an $F$-equilibrium, or it is guaranteed that $f_{ij}>\|s(p,f)\|_{_1}$ for
some $(i,g_j)\in E\setminus F$ connecting two different connected
components of $F$. 
 Such an edge $(i,g_j)$
can be added to $F$ by Lemma~\ref{lem:F-extend}. 

The existence of the edge will follow from the decrease in the potential $\phi(p,f)$ during the algorithm  \textsc{DM}($F,\hat p$); this is proved in Section~\ref{sec:dm}.
\begin{theorem}\label{thm:dm}
There exists a strongly polynomial $O(n^{9}\log^2 n) $ time algorithm, that, for a given
$F\subseteq E$ and prices $\hat p$, computes an $F$-allocation $(p,f)$
such that 
\[
\phi(p,f)\le \frac{\phi(\hat p, \tilde f)}{n^{4}(m+1)},
\]
where $\tilde f$ is the balanced flow computed by \textsc{Balanced}($F,\hat p$).
\end{theorem}
The following lemma shows how the decrease in $\phi(p,f)$ implies the existence of an edge with large flow value connecting two components of $F$.
\begin{lemma}\label{lem:new-edge}
Let $(p,f)$ be an $F$-allocation with
$\phi(p,f)<\Psi(F)/(n(m+1))$. Then, $f_{ij}>\|s(p,f)\|_{_1}$ for at least one edge
$(i,g_j)\in E\setminus F$ such that $i$ and $g_j$ are in two different
undirected connected components of $F$. 
\end{lemma}
\begin{proof}
For a contradiction, assume $f_{ij}\le \|s(p,f)\|_{_1}$
for every $(i,g_j)\in E\setminus F$ where $i$ and $g_j$ are in
different connected components of $F$.
 Let us modify $f$ to another flow  $f'$ as follows. We start by setting $f'=f$; then, for
 every $(i,g_j)\in E\setminus F$ where $i$ and $g_j$ are in the same
 component, we reroute $f_{ij}$ units of flow from $i$ to $g_j$ using a
 path in $F$. Such a path may contain both forward and reverse
 edges; we increase the flow on forward edges and decrease it on
 reverse edges. Further, we set $f'_{ij}:=0$ for $(i,g_j)\in
E\setminus F$ if $i$ and $g_j$ are in different connected components
of $F$. The assumption yields 
\[\|s(p,f')\|_{_\infty}\le \|s(p,f')\|_{_1}\le
(m+1)\|s(p,f)\|_{_1}\le n(m+1) \|s(p,f)\|_{_\infty}.\]
Here, we used that setting $f'_{ij}$ to 0 may increase the surplus by
$f_{ij}\le \|s(p,f)\|_1$ according to the assumption.  
 Therefore,
\[
\phi(p,f')\le n(m+1)\phi(p,f)<\Psi(F),
\]
a contradiction to the definition of $\Psi(F)$, since $\supp(f')\subseteq F$.
\end{proof}

 Using the above, we are
ready to prove Theorem~\ref{thm:main}.

\medskip

\noindent{\it Proof of Theorem~\ref{thm:main}.}
By Lemma~\ref{lem:new-edge}, the number of connected components of
  $F$ decreases after every cycle of Algorithm~\ref{alg:oa}; thus,
  the total number of cycles  is $\le 2n-1$.
Consider any cycle.
Let $(\hat p,\hat f)$ denote that $F$-allocation returned by
\textsc{Boost}($F$) with $\phi(\hat p,\hat f)\le (n-1)^2 \Psi(F)$, and
let $(\hat p, \tilde f)$ denote the balanced $F$-flow at prices $\hat p$. Then,
\[
\|s(\hat p,\tilde f)\|_{_\infty} \le \|s(\hat p,\tilde f)\|_{_1}\le \|s(\hat p, \hat
f)\|_{_1}\le n \|s(\hat p, \hat f)\|_{_\infty},\]
since $(\hat p, \tilde f)$ minimizes $\|s(\hat p,\tilde f)\|_{_1}$ among
all $F$-allocations. Therefore $\phi(\hat p,\tilde f)< n^3 \Psi(F)$. Theorem~\ref{thm:dm}
guarantees that \textsc{DM}($F,\hat p$) finds an $F$-allocation
$(p,f)$ with $\phi(p,f)<\Psi(F)/(n(m+1))$. By Lemma~\ref{lem:new-edge},
$F$ is extended by at least one new edge in this
cycle.
The overall running time estimation is dominated by the running
  time estimation of the calls to \textsc{DM}.\hfill $\Box$\medskip

\section{$F$-allocations and $F$-equilibria}\label{sec:F}
This section is devoted to the proof of the Lemmas~\ref{lem:F-extend}
and \ref{lem:F-eq}, which enable us to add new edges to the revealed
set $F$, as well as to convert an $F$-equilibrium to an equilibrium.
\medskip

\noindent{\it Proof of Lemma~\ref{lem:F-extend}.}
The claim is immediate if $(k,g_\ell)\in F$; for the rest of the proof
we therefore assume $(k,g_\ell)\notin F$.
Let $(p',f')$ be any market equilibrium; we need to show that $(k,g_\ell)\in\MBB(p')$.
 Every edge in $F$ is an
MBB-edge for $(p,f)$ by the definition of an $F$-allocation, and also
for $(p',f')$ because of $F\subseteq F^*$. By the price scaling property of market equilibria, we can assume that $p_\ell' = p_\ell$. Let $T\subseteq G$ be the set of goods whose prices at $p'$ are at least
the prices at $p$, i.e., $T:=\{g_j\in G\ |\ p'_j \ge p_j\}$.
 Clearly, $g_\ell \in T$. Let $S$ be the set of agents who have MBB
edges to goods in $T$ at $p'$, i.e., 
\[S := \{i \in A\ |\ \exists g_j\in T, (i,g_j)\in \MBB(p')\}\enspace .\] 
For a contradiction, suppose $(k,g_\ell)\notin \MBB(p')$. First, we show that $k\not\in
S$. Indeed, if there existed a good $g_j\in T$ such that
$(k,g_j)\in \MBB(p')$, then we would get 
$u_{k\ell}/p_\ell =u_{k\ell}/p'_{\ell}<u_{ij}/p'_{j}\le u_{ij}/p_j$, a
contradiction to $(k,g_\ell)\in \MBB(p)$.

Consider the goods in $T$ at prices $p$. 
 Since no good is oversold, we have
\[
\sum_{i\in A, g_j\in T} f_{ij}\le p(T)\enspace .
\]
Also note that  $f_{ij}\ge 0$ whenever $i\in A\setminus  S$ and
$g_j\in T$. This is because $f_{ij}$ could be negative only on edges
in $F$; however, for every $(i,g_j)\in F$
such that $g_j\in T$, we must have $i\in S$, since $F\subseteq F^*\subseteq \MBB(p')$ by the definition of $F^*$.
Therefore, the previous
inequality implies
\[
\sum_{i\in S, g_j\in T} f_{ij}+f_{kl}\le p(T)\enspace .
\]
Observe that by the choice of $T$, there are no edges $(i,g_j)\in \MBB(p)$ with $i\in S$ and $g_j\in G\setminus T$.
Thus, the first term in the left
hand side is the total money they spend. Since their total surplus is
at most $\|s(p,f)\|_{_1}$, we obtain 
\[p(S)  \le p(T)-f_{k\ell}+\|s(p,f)\|_{_1}\enspace . \]
Using $f_{k\ell} > \|s(p,f)\|_{_1}$ and breaking $p(S)$ into two parts, we can rewrite the above as
\begin{equation}\label{eqn:1}
\sum_{j\in S, g_j\not\in T} p_j < \sum_{j\notin S,g_j\in T}p_{j}\enspace . 
\end{equation}
Let us now examine the market equilibrium $(p',f')$, where $f'\ge 0$
and $f'_{ij} > 0$ is allowed only for the MBB edges at $p'$. Since every agent spends their budget exactly at equilibrium, we have
\[p'(S)=\sum_{i\in S,g_j\in G} f'_{ij}\enspace .
\]
Further, since there are no MBB edges at prices $p'$ from agents outside $S$ to goods in $T$, we obtain 
\[p'(S)= \sum_{i\in S,g_j\in G} f'_{ij}  \ge\sum_{i\in S, g_j\in T} f'_{ij} =  p'(T)\enspace .
\]
Again we can rewrite the above as
\begin{equation}\label{eqn:2}
\sum_{j\in S, g_j\not\in T} p'_j \ge \sum_{j\notin S,g_j\in T}p'_{j}\enspace . 
\end{equation}
Now, since $p'_j \ge p_j$ for $g_j\in T$ and $p'_j < p_j$ for $g_j\not\in T$,
\eqref{eqn:1} and \eqref{eqn:2} give a contradiction. \hfill $\Box$\medskip

We formulate a simple corollary that will be needed in the proof of Lemma~\ref{lem:F-eq}.
\begin{corollary}\label{cor:Fs}
If $F\subseteq F^*$, then for every $F$-equilibrium $f$,
$\supp(f)\subseteq F^*$.
\end{corollary}

\noindent{\it Proof of Lemma~\ref{lem:F-eq}.}
An $F$-equilibrium $(p,f)$ may not be a
market equilibrium since $f$ can have negative values on some edges in
$F$. If that is the case, we find a flow $\tilde f$ such that
$(p,\tilde f)$ is a market equilibrium as follows.

Let us construct the network $N(p)$ on vertex set
$A\cup G\cup \{s, t\}$, where $s$ is a source node and $t$ is a sink
node, and the following set of edges: $(s, i)$ with capacity
$p_i$ for each $i\in A$, $(g_j, t)$ with capacity $p_j$ for each
$g_j\in G$, and $(i, g_j)$ with infinite capacity for each $(i,g_j)\in
\MBB(p)$. Let us use  Orlin's algorithm  \cite{Orlin13} to obtain a
maximum $s-t$ flow $\tilde f$ in $N(p)$ in time $O(nm)$.

 If $\tilde{f}$ has value $\sum_{g_j\in G} p_j$, then clearly $(p,\tilde{f})$ is a market
equilibrium. We show that this must  indeed be the case. 
Let $(p',f')$ be a market equilibrium. By the price scaling property of market equilibria,
we can assume that
$\sum_{g_j\in G} p_j' = \sum_{g_j\in G} p_j$. If $p' = p$, then
clearly $(p,\tilde{f})$ is an equilibrium. For the rest of the proof, assume $p'\neq
p$, let $\bar E := F \cup \supp(f)$, and let  $E' := \MBB(p')$. Corollary~\ref{cor:Fs} and the definition of $F^*$ imply $\bar E
\subseteq F^*\subseteq E'$.  

Let $\alpha := \max_{g_j\in G} \{p_j'/p_j\}$, and $T := \{j\in
G\ |\ p_j'/p_j = \alpha\}$. Let $S$ be the set of agents who
have at least one MBB edge to goods in $T$ when prices are $p'$, i.e., 
\[S = \{i\in A\ |\ \exists g_j\in T, (i,g_j) \in E'\}.\]
We must have
\begin{equation}\label{eqn:atpp}
\sum_{i\in S} p'_i \ge \sum_{g_j \in T} p_j'\enspace . 
\end{equation}

Now consider the connected components of the bipartite graph $(A\cup G, \bar E)$. 
Since $\bar E\subseteq E'$, it follows that $p_j/p_{k}=p'_j/p'_{k}$ whenever $j$ and $k$ are in the same connected component. Thus, if $T\cap C_i\neq \emptyset$ for a connected component, then $G\cap C_i\subseteq T$.
This implies that $S\cup T$ is the union of some connected components $C_1, \dots, C_\ell$ of $(A\cup G, \bar E)$, that is, 
$\bigcup_{k=1}^\ell (G\cap C_k) = T$, and 
$\bigcup_{k=1}^\ell (A\cap C_k) = S$.
 At the $F$-equilibrium $(p,f)$, we have 
\begin{equation}\label{eqn:atp}
\sum_{i\in S} p_i = \sum_{g_j \in T} p_j\enspace . 
\end{equation}
The equations \eqref{eqn:atpp} and \eqref{eqn:atp}, along with the
definition of $\alpha$,  imply that
\[
\alpha\sum_{i\in S} p_i \ge \sum_{i\in S} p'_i \ge
\sum_{g_j \in T} p_j' =\alpha \sum_{g_j \in T} p_j=\alpha \sum_{i\in
  S} p_i \enspace .
\]
Consequently, 
$\sum_{i\in S} p'_i = \sum_{g_j \in T} p_j'$ and $T = \{g_j\in
G\ |\ j\in S\}$, i.e., the set of goods brought by agents in
$S$ is exactly equal to $T$. This further implies that $f'$
for agents in $S$ and for goods in $T$ is supported only on the $\MBB(p)$ edges between them, and hence $\tilde{f}$ must saturate all agents in $S$ (and equivalently, all goods in $T$) because the set of MBB edges between $S$ and $T$ remains same at $p$ and $p'$. Next, we remove the agents in $S$ and goods in $T$, and repeat the same analysis on the remaining set of agents and goods. This proves that $(p,\tilde{f})$ is a market equilibrium. \hfill$\Box$\medskip

\section{The Duan-Mehlhorn (DM) subroutine}\label{sec:dm}

\begin{algorithm}[t]
\caption{DM$(F,\hat p)$}\label{alg:dm}
\DontPrintSemicolon
\SetKwInOut{Input}{Input}\SetKwInOut{Output}{Output}
\Input{Utilities $(u_{ij})_{i\in A, g_j\in G}$, an edge set $F\subseteq E$, and prices $\hat p$ with $F\subseteq\MBB(\hat p)$.} 
\Output{An $F$-equilibrium $(p,f)$ or an $F$-allocation $(p,f)$ such that $f_{ij} > \|s(p,f)\|_{_1}$ for an $(i,g_j)\in E\setminus F$, where $i$ and $g_j$ are in different connected components of $F$.}
$p\ot \hat p$;
$f\ot$ \textsc{Balanced}$(F,p)$\tcp*{Lemma~\ref{lem:balanced}}
\Repeat{either $f_{ij} > \|s(p,f)\|_{_1}$ for an edge $ (i,g_j) \in E\setminus F$ with $i$ and $g_j$ in different components of $F$, or $\|s(p,f)\|_{_1} = 0$}{
Sort the agents in decreasing order of surplus, i.e., $c_1(p,f)\ge c_2(p,f)\ge\ldots\ge c_n(p,f)$\; 
Find the smallest $\ell$ for which $c_\ell(p,f)/c_{\ell+1}(p,f) > 1+1/n$, and let $\ell=n$ when there is no such $\ell$.\;
$S\ot [\ell]$;\ \ \ \ $\Gamma(S) = \{g_j \in G\ |\ \exists i\in S: f_{ij} \neq 0\}$\;
$\gamma \ot 1$\;
\Repeat{Event 2 or 3 occurs}{
$x \ot 1$; Define $p_j \ot xp_j, \forall g_j\in \Gamma(S), f_{ij} \ot xf_{ij}, \forall i\in S, \forall g_j\in \Gamma(S)$\;\label{alg:pfupdate}\tcp*{$c_i(p,f)$ and $s_j(p,f)$ change accordingly}
Increase $x$ continuously up from $1$ until one of the following events occurs\;
\ \ \ \ \ \ {\bf Event 1:} A new edge, say $(a,g_b)$, becomes MBB\tcp*{$a\in S, g_b\not\in \Gamma(S)$}\label{alg:e1}
\ \ \ \ \ \ {\bf Event 2:} $\min_{i\in S} c_i(p, f) = \max\{\max_{i\not\in S} c_i(p, f), 0\}$\tcp*{Balancing}
\ \ \ \ \ \ {\bf Event 3:} $\gamma x = 1+ 1/(56e^2n^3)$\tcp*{Price-rise}
\If{Event 1 occurs}{
	$\tilde{c}_i(p,f) \ot c_i(p,f),\ \forall i \in S\setminus\{a\}$\;
	$\tilde{c}_a(p,f) \ot c_a(p,f) - p_b$\;
	$\tilde{c}_i(p,f) \ot c_i(p,f) + f_{ib},\ \forall i\notin S$\;
	\If{$\exists i\in A\setminus S$ s.t. $(i,g_b)\in F$ or $\min_{i\in S}\tilde{c}_i(p,f) \le \max\{\max_{i\notin S}\tilde{c}_i(p,f), 0\}$}{
		break\tcp*{break from the inner loop}\label{alg:break}
	}
	$f_{ib}\ot 0, \forall i\in A;\ \ f_{ab} = p_b; \ \ \Gamma(S) \ot \Gamma(S) \cup \{g_b\}; \ \ \gamma \ot \gamma x$\;\label{alg:update}
}
}
	$f\ot$ \textsc{Balanced}($F,p$)\label{alg:bf}
}
\Return{$(p,f)$}\;
\end{algorithm}
In this section, we present a variant of the  Duan-Mehlhorn (DM) algorithm \cite{DuanM15} algorithm as a subroutine DM$(F,\hat p)$ in Algorithm \ref{alg:dm}. 
The input is a revealed edge set $F$ and prices $\hat p$ such that $F\subseteq \MBB(\hat p)$, and the  output is either an $F$-equilibrium, or an $F$-allocation $(p,f)$ where $f_{ij} > \|s(p,f)\|_{_1}$ for some $(i,g_j)\in E\setminus F$ connecting two different components of $F$. The modifications compared to the original DM algorithm are listed in Remark~\ref{rem:difference}. We now provide a description where the subroutine terminates once an arc with $f_{ij}>\|s(p,f)\|_{_1}$ is identified. The variant as required in Theorem~\ref{thm:dm} can be obtained by simply by removing the termination condition, and letting the algorithm run for $O(n^{6}\log^2{n})$ iterations of the outer loop. 

We call one execution of the outer loop a \emph{phase}, and one execution of the inner loop an \emph{iteration}. Algorithm~\ref{alg:dm} first computes a balanced flow $f$ using the subroutine \textsc{Balanced}($F,p$) as in Lemma~\ref{lem:balanced}. Then,   the agents are sorted in decreasing order of surplus. Without loss of generality, we assume that $c_1(p,f) \ge \dots \ge c_n(p,f)$. Then, we find the smallest $\ell$ for which the ratio $c_\ell(p,f)/c_{\ell+1}(p,f)$ is more than $1+1/n$. If there is no such $\ell$ then we let $\ell:=n$. Let $S$ be the set of first $\ell$ agents, and let $\Gamma(S)$ be the set of goods for which there is a non-zero flow from agents in $S$. Since $f$ is balanced, the agents outside $S$ have zero flow to goods in $\Gamma(S)$, i.e., $f_{ij} = 0, \forall i\not\in S, g_j\in \Gamma(S)$ and the surplus of every good in $\Gamma(S)$ is zero. The parameter $\gamma$ measures the cumulative price increment throughout a phase; we set $\gamma=1$ before starting the sequence of inner loops.

Next, the algorithm runs the inner loop where it increases the prices of goods in $\Gamma(S)$ and the flow between agents in $S$ and goods in $\Gamma(S)$ by a multiplicative factor $x\ge 1$ until one of three events occurs. Observe that except for the MBB edges $(i,g_j)$ where $i\notin S, g_j\in \Gamma(S)$, all MBB edges remain MBB with this price change, and the surplus of every good in $\Gamma(S)$ remains zero. When prices of goods in $\Gamma(S)$ increase, an edge $(i,g_j)$ from $i\in S$ and $g_j\not\in \Gamma(S)$ can become MBB. We need to stop when such an event occurs in order to maintain an $F$-allocation; this is captured by Event 1. In Event 2, we stop when the surplus of an agent $i \in S$ becomes equal to either the surplus of an agent $i'\not\in S$ or zero. Let us note that $c_i(p,f)\ge 0$ is maintained throughout; we use the expression $\max\{\max_{i\notin S}c_i(p,f),0\}$ to also cover the possible case $S=[n]$. In Event 3, we stop when $\gamma x$ becomes $1+1/(56e^2n^3)$. 

If Event 1 occurs, then we have a new MBB edge $(a, g_b)$ from $a\in S$ to $g_b \not\in \Gamma(S)$. Using this new edge, it is now possible to decrease the surplus of agent $a$ and increase the surpluses of agents $i\not\in S$ by increasing $f_{ab}$ and decreasing $f_{ib}$. We next check if this can lead to making the surplus of an agent $i\in S$ and $i'\notin S$ equal. Observe that it is always possible if there exists an edge $(i',g_b)\in F$. If yes, then we break from the inner loop, otherwise we update flow so that agent $a$ buys the entire good $g_b$, add $g_b$ to $\Gamma(S)$, update $\gamma$ to $\gamma x$, and go for another iteration. 

\begin{lemma}\label{lem:numIter}
The number of iterations in a phase is at most $n$.
\end{lemma}
\begin{proof}
Consider the iterations of a phase. At the beginning of every iteration, the size of $\Gamma(S)$ grows by $1$, and hence there cannot be more than $n$ iterations in a phase. 
\end{proof}

When we break from the inner loop, we recompute a balanced flow and then check if either $\|s(p,f)\|_{_1}$ is zero or there is an edge $(i,g_j)\notin F$ with $f_{ij} > \|s(p,f)\|_{_1}$ connecting two different components of $F$. If yes, then we return the current $(p,f)$, otherwise we go for another phase. Next, we show that $(p,f)$ remains an $F$-allocation throughout the algorithm, which implies that the algorithm returns an $F$-allocation.

\begin{lemma}
The output $(p,f)$ of Algorithm~\ref{alg:dm} is an $F$-allocation.
\end{lemma}

\begin{proof}
We only need to show that $F\subseteq$ MBB$(p)$ throughout the algorithm. Observe that an MBB edge $(i,g_j)$ becomes non-MBB only if $i\notin S$ and $g_j\in \Gamma(S)$, where $S$ and $\Gamma(S)$ are obtained with respect to a balanced flow $f$. If an edge $(i,g_j)\in F$ is such that $i\notin S$ and $g_j\in \Gamma(S)$ then it contradicts that $f$ is a balanced flow because the edges in $F$ are allowed to carry negative flow. 
\end{proof}

The running time analysis of Algorithm~\ref{alg:dm} is based on the evolution of the norm $\|c(p,f)\|_{_2}$ and prices $p$. 
If a phase terminates due to Event 3, then we call it \emph{price-rise}, otherwise \emph{balancing}. 
The next two lemmas are crucial that eventually imply that the potential function $\phi(p,f)$ decreases substantially within a strongly polynomial number of phases. 

\begin{lemma}\label{lem:mono}
In Algorithm~\ref{alg:dm}, the price of every good monotonically increases and the total surplus, i.e., $\|s(p,f)\|_{_1}$, monotonically decreases.
\end{lemma}
\begin{proof}
Clearly, the price of every good monotonically increases in Algorithm~\ref{alg:dm}. During a price increase step, $s_j(p,f)=0$ is maintained for every $g_j\in \Gamma(S)$, and $s_j(p,f)$ does not change for $g_j\in G\setminus \Gamma(S)$. If the allocation changes during Event 1, then $s_b(p,f)$ decreases to 0, and the other surpluses remain unchanged. 
When a balanced flow is recomputed at the end of a phase, then $\|s(p,f)\|_{_1}$ can only decrease.
\end{proof}

The proof of the next lemma is an adaptation of the proof in~\cite{DuanGM16}, and is given in Appendix~\ref{sec:dm-mp}. 
\begin{restatable}{lemma}{normProgress}\label{lem:progress}
Let $f$ be a balanced flow at the beginning of a phase, and let $(p', f')$ be the prices and flow at the end of the phase. Then
\begin{enumerate}[(i)]
\item $\prod_{j=1}^n p'_j \ge \left(1+\frac1{56e^2 n^3}\right)\prod_{j=1}^n p_j$ in a price-rise phase, and
\item $\|c(p',f')\|_{_2} \le \|c(p,f)\|_{_2}/\left(1+\frac1{56e^2 n^3}\right)$ in a balancing phase.
\end{enumerate}
\end{restatable}

\begin{lemma}
The number of arithmetic operations in a phase of Algorithm~\ref{alg:dm} is $O(n^3)$. 
\end{lemma}
\begin{proof}
From Lemma~\ref{lem:numIter}, the number of iterations in a phase is at most $n$. In each iteration, we find the minimum $x$ where one of the events occur, which takes at most $O(n^2)$ arithmetic operations. If Event $1$ occurs, then we define another surplus vector $\tilde{c}(p,f)$, and based on this we decide to exit from the inner loop. This requires additional $O(n^2)$ arithmetic operations. In total, each iteration takes $O(n^2)$ arithmetic operations. The steps before the inner loop like sorting etc. takes $O(n\log{n})$ arithmetic operations. We compute a balanced flow after exiting from the inner loop, in time $O(nm\log{(n^2/m)})$ according to Lemma~\ref{lem:balanced}. Overall, each phase takes $O(n^3)$ arithmetic operations since $m\le n^2$. 
\end{proof}

In the next lemma, we show that the potential function $\phi(p,f)$ decreases by a large factor within a strongly polynomial number of phases. This together with Lemmas~\ref{lem:F-extend} and \ref{lem:new-edge}  imply that every major cycle terminates in strongly polynomial time.  

\begin{lemma}
The potential function $\phi(p,f)$ decreases by a factor of at least $1/n^{\gamma}$ in $O((2+\gamma)^2n^{6}\log^2{n})$ phases of Algorithm~\ref{alg:dm} for any $\gamma>0$.
\end{lemma}
\begin{proof}
Every phase of Algorithm~\ref{alg:dm} is either \emph{price-rise} or \emph{balancing}. Using $\|s(p,f)\|_{_1}/n \le \|s(p,f)\|_{_\infty} \le \|s(p,f)\|_{_1}$, we have the following inequality:
\begin{equation}\label{eqn:norm}
\frac{\|s(p,f)\|_{_1}}{n(\prod_{j} p_j)^{1/n}} \le \phi(p,f) = \frac{\|s(p,f)\|_{_\infty}}{(\prod_{j} p_j)^{1/n}} \le \frac{\|s(p,f)\|_{_1}}{(\prod_{j} p_j)^{1/n}}\enspace .
\end{equation}
Recall from Lemma~\ref{lem:mono} that $\|c(p,f)\|_1$ is monotonically decreasing and the prices are monotonically increasing throughout.
According to Lemma~\ref{lem:progress}, if there are $C(2+\gamma)n^3\log n$ consecutive balancing phases for some constant $C>0$, then $\|c(p,f)\|_{_2}$ decreases by a factor of at least $\nicefrac{1}{n^{2+\gamma}}$.  This further implies that the $\ell_1$ norm, i.e., $\|s(p,f)\|_{_1}=\|c(p,f)\|_1$, decreases by a factor of at least $\nicefrac{1}{n^{1.5+\gamma}}$.

Consider now a sequence of $C^2(2+\gamma)^2 n^{6}\log{n}$ phases. If this contains 
$C(2+\gamma)n^3\log n$ consecutive balancing phases, then the statement follows as above. Otherwise,
there are at least $C(2+\gamma)n^3\log n$  price-rise phases. In that case, the geometric mean of prices, i.e., $(\prod_j p_j)^{1/n}$, increases by a factor of at least $n^{2+\gamma}$. This together with \eqref{eqn:norm} proves the claim.
\end{proof}

\noindent{\it Proof of Theorem~\ref{thm:dm}.}
According to the above lemma, if we do not terminate Algorithm~\ref{alg:dm} in the first iteration when an arc $(i,g_j)\in E\setminus F$ with $f_{ij}>\|s(p,f)\|_{_1}$ is found, then the potential $\phi(p,f)$ decreases by a factor $n^{4}(m+1)$ within $O(n^{6}\log^2{n})$ phases.

For a strongly polynomial algorithm, we also need to keep all intermediate numbers polynomial bit length. For this, we can use the Duan-Mehlhorn \cite{DuanM15} technique by restricting the prices and update factor $x$ to powers of $(1+1/L)$ where $L$ has polynomial bit length. This guarantees that all arithmetic is performed on rational numbers of polynomial bit length. As shown in~\cite{DuanM15} this does not change the number of iterations of the DM subroutine. \hfill$\Box$\medskip

\begin{remark}\label{rem:difference}
Compared to the original DM algorithm in \cite{DuanM15}, Algorithm~\ref{alg:dm} differs in the following.
\begin{enumerate}
\item We handle Event 1 (in line \ref{alg:e1}) differently than the other two events and this gives rise to two nested loops, unlike~\cite{DuanM15} where every event is handled similarly and there is only one loop. 
\item The edges in $F$ are allowed to carry negative flow, unlike~\cite{DuanM15} where flow is always non-negative. 
\item We initialize prices to $p$, unlike~\cite{DuanM15} where every price is initialized to $1$. And, we stop once a new edge is revealed.
\end{enumerate}
\end{remark}

\section{A linear program for $\Psi(F)$}\label{sec:boost}
In this section, we first formulate an LP to
compute $\Psi(F)$. Then, we introduce the class of $\Z$-matrices,
and formulate a general statement (Theorem~\ref{thm:Z0}) that shows
how certain LPs with a $\Z$ constraint matrix can be approximated by a
 two variable per inequality system. We use this to prove Theorem~\ref{thm:boost}.
The proof of Theorem~\ref{thm:Z0} will be
given in Section~\ref{sec:approximate}.

Given $F\subseteq E$, we consider the bipartite graph $(A\cup G,
F)$. Let 
$C_1,C_2,\ldots,C_t$ denote the connected components that have a
non\-empty intersection with $G$. (In particular, we include all
isolated vertices in $G$, but not those in $A$.) Let
$\gamma_i:=|C_i\cap G|$.
Let us fix an arbitrary good in each of these components; for
simplicity of notation, let us assume that the fixed good in $C_i$ is $g_i$.

If all edges in $F$ are forced to be MBB edges, then fixing the price
$p_i$ of $g_i$ uniquely determines the prices of all goods in $C_i\cap
G$. Indeed, for any buyer $k\in C_i\cap A$, and any goods
$g_{\ell},g_{\ell'}\in C_i\cap G$ with $k\ell,k\ell'\in F$, we have that
$p_\ell/p_{\ell'}=u_{k\ell}/u_{k\ell'}$. 
Consequently, for any $i\in [t]$,  and for any $g_\ell\in C_i\cap G$,
we can compute the multiplier
$\theta_{i\ell}>0$ such that $p_\ell=\theta_{i\ell}p_i$ whenever all
edges in $F$ are MBB. 
For an agent $\ell\in A$,
let $\rho(\ell)\in [t]$ denote the index of the component containing the
good $g_\ell$ of this agent:
that is, $g_\ell\in C_{\rho(\ell)}\cap G$, and
$p_\ell=\theta_{\rho(\ell)\ell}p_{\rho(\ell)}$.
Let $\Theta_i:=\sum_{g_\ell\in C_i\cap G} \theta_{i\ell}$; the
total price of the goods in $C_i$ is $\Theta_i p_i$.

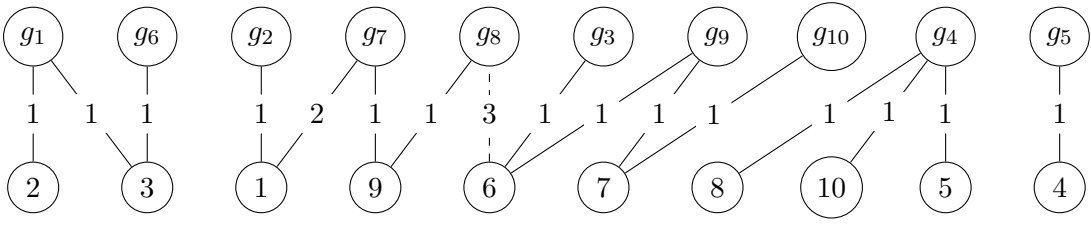
\begin{figure*}[htp]
\begin{center}
\begin{tikzpicture}
    \node[shape=circle,draw=black] (G1) at (0,2) {$g_1$};
    \node[shape=circle,draw=black] (G6) at (1.5,2) {$g_6$};
    \node[shape=circle,draw=black] (G2) at (3,2) {$g_2$};
    \node[shape=circle,draw=black] (G7) at (4.5,2) {$g_7$};
    \node[shape=circle,draw=black] (G8) at (6,2) {$g_8$};
    \node[shape=circle,draw=black] (G3) at (7.5,2) {$g_3$};
    \node[shape=circle,draw=black] (G9) at (9,2) {$g_9$};
    \node[shape=circle,draw=black] (G10) at (10.5,2) {$g_{10}$};
    \node[shape=circle,draw=black] (G4) at (12,2) {$g_4$};
   \node[shape=circle,draw=black] (G5) at (13.5,2) {$g_5$};
    \node[shape=circle,draw=black] (A2) at (0,0) {$2$};
    \node[shape=circle,draw=black] (A3) at (1.5,0) {$3$};
    \node[shape=circle,draw=black] (A1) at (3,0) {$1$};
    \node[shape=circle,draw=black] (A9) at (4.5,0) {$9$};

    \node[shape=circle,draw=black] (A6) at (6,0) {$6$};
    \node[shape=circle,draw=black] (A7) at (7.5,0) {$7$};
   \node[shape=circle,draw=black] (A8) at (9,0) {$8$};
    \node[shape=circle,draw=black] (A10) at (10.5,0) {$10$};
  \node[shape=circle,draw=black] (A5) at (12,0) {$5$}; 
  \node[shape=circle,draw=black] (A4) at (13.5,0) {$4$}; 
  \tikzset{LabelStyle/.style =   {draw=white,
                                  fill           = white,         
                                  text           = black}}
\begin{scope}[>={Stealth[black]}]
 \draw (A1) to node[LabelStyle]{$1$} (G2);
 \draw (A1) to node[LabelStyle]{$2$} (G7);
 \draw (A9) to node[LabelStyle]{$1$} (G7);
 \draw (A9) to node[LabelStyle]{$1$} (G8);
 \draw (A6) to node[LabelStyle]{$1$} (G3);
 \draw (A6) to node[LabelStyle]{$1$} (G9);
 \draw (A7) to node[LabelStyle]{$1$} (G9);
 \draw (A7) to node[LabelStyle]{$1$} (G10);
 \draw (A8) to node[LabelStyle]{$1$} (G4);
 \draw (A10) to node[LabelStyle]{$1$} (G4);
 \draw (A5) to node[LabelStyle]{$1$} (G4);
 \draw (A4) to node[LabelStyle]{$1$} (G5);
 \draw[dashed] (A6) to node[LabelStyle]{$3$} (G8);
\end{scope}
\end{tikzpicture}
\end{center}
\caption{Example problem setting. }\label{fig:boost-1} 
\end{figure*}

\subsection{Constructing the LP}\label{sec:reduce-LP}

The variables $(p_i)_{i\in [t]}$ uniquely determine the price of every good.
We compute $\Psi(F)$ in terms of these variables. To differentiate between this
$t$-dimensional price
vector and the $n$-dimensional price vector of all goods, we say that for a price
vector $\bp\in \R^t$, the vector $p\in \R^n$ is the
\emph{extension} of $\bar p$, if $p_\ell=\theta_{\rho(\ell)\ell}\bp_{\rho(\ell)}$ for all $\ell\in [n]$ (in particular, $p_\ell=\bp_\ell$
for $\ell\in [t]$). We also say that the $F$-allocation $(p,f)$ is an
extension of $\bp$, if $p$ is the extension of $\bar p$.

\examplebox{
{\it Example.} Throughout this  section and the next one, we  illustrate the
argument with the example in Figure~\ref{fig:boost-1}. There are 10
agents and 15 edges in $F$. The edges in $F$ are depicted by solid
edges with the $u_{ij}$ values indicated; all these are 1 except for
$u_{17}=2$. The picture does not include the edges in $E\setminus F$ except for
one example: the
dashed line for $(6,g_8)$ with $u_{68}=3$.
There are 5 connected
components, containing goods $\{g_1,g_6\}, \{g_2,g_7,g_8\}, \{g_3,g_9,g_{10}\}, \{g_4\}$,
and $\{g_5\}$, with $p_6=p_1$, $p_7=p_8=2p_2$, and $p_3=p_9=p_{10}$. Thus,
$\Theta_1=2$, $\Theta_2=5$, $\Theta_3=3$, $\Theta_4=1$, $\Theta_5=1$, and
$\gamma_1=2$, $\gamma_2=3$, $\gamma_3=3$, $\gamma_4=1$, $\gamma_5=1$.
}\medskip\medskip

We now formulate linear constraints that ensure that a vector $\bp\in \R^t$ 
can be extended to an
$F$-allocation $(p,f)$ with $\|s(p,f)\|_{_\infty}\le 1$, and
$\supp(f)\subseteq F$. The first set of constraints will enforce that all edges in $F$ are MBB,
and the second set will guarantee the existence of a desired money
flow $f$ with the surplus bounds.

First, the edges in $F$ are MBB if and only if
$u_{kj}/p_j\le u_{kj'}/p_{j'}$ for any $k\in A$, and any
$g_j,g_{j'}\in G$ such that $(k,g_j)\in E$, $(k,g_{j'})\in F$. 
The $\theta_{i\ell}$ coefficients already capture that equality holds
if $(k,g_j),(k,g_{j'})\in F$. For the rest of the pairs,
we can express this
constraint in terms of the $\bp$ variables as 
\begin{equation}\label{eq:comp-mbb}
\begin{aligned}
u_{kj}\theta_{\rho(j')j'}\bp_{\rho(j')} -
u_{kj'}\theta_{\rho(j)j}\bp_{\rho(j)} & \le 0\quad  &  \forall k,j,j'\in A,\, (k,& g_j)  \in E\setminus F,\, (k,g_{j'})\in F.
\end{aligned}
\end{equation}
We add a second set of constraints for $\|s(p,f)\|_{_\infty}\le 1$. Since
$f$ is supported on $F$ and is allowed to be negative, this can be
guaranteed if and only if for any component $C_i$, $i\in [t]$, the total price of the goods in $C_i\cap G$ exceeds the
total budget of the agents in $C_i\cap A$ by at most
$\gamma_i=|C_i\cap G|$. Recall
that given the prices $\bp$ of the fixed goods, the total
price of goods in $C_i\cap G$ is  $\Theta_i \bp_i$.
We obtain the constraints
\begin{equation}\label{eq:comp-balance}
\Theta_i \bp_i- \sum_{k\in C_i\cap A} \theta_{\rho(k)k} \bp_{\rho(k)}\le \gamma_i
\quad\forall i\in[t].
\end{equation}
Let us now define the following LP:
\begin{equation}\label{LP:PF}\tag{$P_F$}
\begin{aligned}
\max& \sum_{i=1}^t \bp_i \\
\mbox{s. t. constraint sets }&\mbox{\eqref{eq:comp-mbb} and \eqref{eq:comp-balance}},\\
\bp&\ge 0.
\end{aligned}
\end{equation}
Note that $\bp=0$ is a feasible solution. Using LP duality, the above
program is unbounded if and only if the next LP has a feasible
solution $\bp\neq 0$.
\begin{equation}\label{LP:PF-0}\tag{$P_F^0$}
\begin{aligned}
\mbox{constraint  set }& \eqref{eq:comp-mbb}, \\
\Theta_i \bp_i- \sum_{k\in C_i\cap A} \theta_{\rho(k) k} \bp_{\rho(k)}&\le 0
\quad\forall i\in[t]\\
\bp&\ge 0.
\end{aligned}
\end{equation}

\examplebox{
{\it Example.}
Let us show the formulation for the example in
Figure~\ref{fig:boost-1}. The variables are
$\bp_1,\bp_2,\bp_3$,  $\bp_4$, and $\bp_5$. From the constraint set \eqref{eq:comp-mbb},
we only show the example of $k=6$, $j=8$, and $j'=3$:
\[
3\bp_3-2\bp_2\le 0.
\]
For the components, we have
\[
\begin{aligned}
2\bp_1-\bp_2-\bp_3&\le 2\\
5\bp_2-\bp_1-\bp_3&\le 3\\
3\bp_3-\bp_1-2\bp_2&\le 3\\
\bp_4-2\bp_2-\bp_3-\bp_5&\le 1\\
\bp_5-\bp_4&\le 1.
\end{aligned}
\]
}\medskip\medskip

\begin{lemma}\label{lem:point-max}
\begin{enumerate}[(i)]
\item Any solution $\bp\in\R^n$ to \eqref{LP:PF} can be extended to an 
$F$-allocation $(p,f)$ with $\|s(p,f)\|_{_\infty}\le 1$, and
$\supp(f)\subseteq F$. 
\item If \eqref{LP:PF} is bounded, then
there exists a  pointwise maximal solution $\bp^*\in \R^t$, that is, $\bp\le\bp^*$ for
any solution $\bp\in\R^t$ to \eqref{LP:PF}.
Let $(p^*,f^*)$ denote
the extension of these prices to an 
$F$-allocation  with $\|s(p^*,f^*)\|_{_\infty}\le 1$, and
$\supp(f^*)\subseteq F$.  Then, $\Psi(F)=\phi(p^*,f^*)$.
\item Under assumption \eqref{eq:strongly}, every nonzero solution to \eqref{LP:PF-0}  is strictly
  positive. Such a solution can be extended to an $F$-equilibrium.
\end{enumerate}
\end{lemma}
\begin{proof}
The proof of part {\em (i)} was given above. For part {\em (ii)},
let $\bp$ and
$\bp'$ be two different solutions, and let
$\bp''_i=\max\{\bp_i,\bp'_i\}$ for all $i\in [t]$. Then it is easy to see
that $\bp''$ satisfies all inequalities 
\eqref{eq:comp-mbb} and
\eqref{eq:comp-balance}.
(This is true more generally for systems where the transposed of the constraint matrix is pre-Leontief, see Section~\ref{sec:M2VPI}. )
  This implies the existence of a pointwise maximal $\bp^*$ in case \eqref{LP:PF} is bounded.

Let $(p,f)$ be the optimal
$F$-allocation such that $\Psi(F)=\phi(p,f)$. Since
$\phi(p,f)=\phi(\alpha p,\alpha f)$ for any $\alpha>0$, we may assume
that $\|s(p,f)\|_{_\infty}=1$. Then, $(p_i)_{i\in [t]}$ is  feasible to
\eqref{LP:PF}, and therefore  $p_i\le p^*_i$ for all $i\in
[t]$. Consequently, $\phi(p^*,f^*)\le \phi(p,f)$. By definition of
$\Psi(F)$, we have 
$\phi(p^*,f^*)\ge \Psi(F)$. Thus, equality must hold, which in particular
implies $p_i^*=p_i$ for all $i\in [t]$. 

Let us now turn to part {\em (iii)}, when \eqref{LP:PF-0} has a
nonzero solution $\bp\in \R^t$ (and thus \eqref{LP:PF} is
unbounded). Let $p\in \R^n$ be the extension of $\bp$.

Summing up the second set of constraints in \eqref{LP:PF-0}, the
coefficient of every $\bp_i$ is nonnegative. Consequently, all these
inequalities must hold at equality, implying that 
 in every component $C_i$, the total budget of the agents in $C_i\cap
 A$ equals the total price of the goods in $C_i\cap G$.  Further, if $i\in A$ has no
incident edges in $F$, then $p_i=0$ must hold. 
If  $p>0$, then the union of the components $C_i$ equals $A\cup G$. Then,
 we
can set
$f$ with $\supp(f)\subseteq F$ such that $\|s(p,f)\|_{_1}=0$, and such
 a $(p,f)$ gives an $F$-equilibrium.

It is left to show that $\bp>0$, or equivalently, $p>0$. Let $A_0:=\{k\in A:\, p_k=0\}$, and
assume $A_0\neq\emptyset$. Note that $A\setminus A_0\neq \emptyset$
since $p\neq 0$.
By  assumption \eqref{eq:strongly}, there exist $k\in
A\setminus A_0$ and $\ell\in A_0$ such that $u_{k\ell}>0$. 
There exists at least one edge $(k,g_j)\in F$, as otherwise $p_k=0$, as
shown above. By  \eqref{eq:comp-mbb}, $j\in A_0$ must hold. 
Let $C_i$ be the component containing agent $k$ and good $g_j$.
Clearly, for every $g_{j'}\in
C_i\cap G$, we have $j'\in A_0$. Hence, the budget of every agent in
$C_i\cap A$ must also equal 0, in particular, $p_k=0$, a
contradiction.
\end{proof}

\subsection{Monotone two variable per inequality systems}\label{sec:M2VPI}
Let ${\cal M}_2(m,n)$ denote the set of $m\times n$ rational matrices
with at most one positive and at most one
negative entry per row. For a matrix $A\in {\cal M}_2(m,n)$, and an arbitrary
vector $b\in \mathbb{Q}^m$, the LP $Ax\le b$ is called a \emph{monotone two variable per inequality
  system}, abbreviated as \emph{M2VPI}. 
In every such system, whenever the objective function $\max \sum_i x_i$ is bounded, there exists a \emph{pointwise
maximal feasible solution}, that is, a feasible $x^*$ such that for every
feasible solution $x$, $x\le x^*$.

This property holds more generally. Namely, a matrix is called
\emph{pre-Leontief} if every column contains at most one positive
element. If $A^\top$ is pre-Leontief, then the system $A^\top x\le c$
has a pointwise maximal feasible solution whenever $\max \sum_i x_i$
is bounded; see \cite{Cottle72}. 
Whereas M2VPI systems are strongly polynomially solvable, as stated in
the next theorem, no such algorithm
is currently known for the general pre-Leontief setting.

\begin{theorem}[\cite{Megiddo83,HochbaumN94}] \label{thm:2v}
Consider an M2VPI system $Ax\le b$ with  $A\in {\cal M}_2(m,n)$. Then
there exists a strongly polynomial  $O(mn^2 \log m)$
time algorithm that finds a feasible
solution or concludes infeasibility. Further, if there exists a pointwise
maximal feasible solution, the algorithm also finds that one.
\end{theorem}
Note that this theorem is not directly applicable to
\eqref{LP:PF}. Whereas the constraints
\eqref{eq:comp-mbb} are of the required form, the constraints
\eqref{eq:comp-balance} have only one positive coefficient but
possibly multiple negative ones. In what follows, we show that finding
an approximate solution to \eqref{LP:PF}
can be reduced to an M2VPI system.

\subsection{$\Z$-matrices}
Let $M\in \R^{t\times t}$ be the matrix representing the left hand
side of the constraints in \eqref{eq:comp-balance}. That is, 
for all $i,j\in [t]$, we let
\begin{equation}\label{eq:M-def}
M_{ij}:=\begin{cases}
\Theta_i -\sum_{k\in C_i\cap A: \rho(k)=i} \theta_{ik}, &\mbox{if
  }i=j,\\
-\sum_{k\in C_i\cap A: \rho(k)=j} \theta_{jk}, &\mbox{if }i \neq j.
\end{cases}
\end{equation}
Thus, \eqref{eq:comp-balance} can be written as $M\bp\le \gamma$,
where $\gamma^\top=(\gamma_1,\ldots,\gamma_t)$.

\begin{definition} A matrix $M\in \Q^{k\times t}$ is a $\Z$-matrix,
  if all off-diagonal entries are nonpositive,
  and all column sums are nonnegative.  
 For a non-square matrix, by diagonal entries we mean all entries $z_{ii}$ for $1\le i\le\min\{k,t\}$. 
  We let ${\cal \Z}(k,t)$ denote the
set of $k\times t$ $\Z$-matrices.
\end{definition}

Clearly, the matrix $M$ defined by \eqref{eq:M-def} is in ${\cal
  \Z}(t,t)$. Recall that a matrix is called a $Z$-matrix if all off-diagonal
entries are nonpositive; the notation reflects the additional
requirement on the columns. Further, note that a matrix is a
$\Z(t,t)$-matrix if and only if it is the transposed of a weighted
Laplacian of a directed graph on $t$ vertices, or if it can be
obtained by deleting a row and a column of the transposed of a weighted
Laplacian of a directed graph on $t+1$ vertices.
We will prove the following theorem on LPs with $\Z$-matrices as
constraint matrices.
\begin{theorem}\label{thm:Z0}
Given a matrix $M\in {\cal \Z}(k,t)$ with $\ell$ nonzero entries, and
$b\in \Q^k$, with $b>0$, we let 
\[
P_M=\{x\in \R^t:\, Mx\le b, x\ge 0 \}.
\]
Then, in time $O(\ell(\min\{k,t\})^3)$, we can construct a matrix $\bar M\in {\cal
  M}_2(\ell',t)$ and $\bar b\in \Q^{l'}$ for $\ell'\le \ell$ such that 
\[
P_M\subseteq \{x\in \R^t:\, \bar Mx\le \bar b, x\ge 0 \}\subseteq B^2 P_M,
\]
where $B=\sum_{j=1}^k b_j/\min_{i\in[k]} b_i$. Further, the size of the entries in $\bar M$ and $\bar b$ are polynomially bounded in the encoding size of the input.
\end{theorem}
Here, we use the notation $\alpha P=\{\alpha x:\, x\in P\}$ for a set
$P$ and a constant $\alpha>0$. The proof of Theorem~\ref{thm:Z0} will
be given in Section~\ref{sec:approximate};
we now use it to derive Theorem~\ref{thm:boost}.\medskip

\noindent{\it Proof of Theorem~\ref{thm:boost}.}
Lemma~\ref{lem:point-max} establishes that computing $\Psi(F)$ is
equivalent to solving the LP \eqref{LP:PF}. We construct a second LP
$Q_F$ as follows. For the
constraint set \eqref{eq:comp-balance} in the form $M\bp\le \gamma$,
we apply Theorem~\ref{thm:Z0} to obtain $\bar M\bp\le \bar
\gamma$. Note that $B\le n$, since $\sum_{i=1}^t \gamma_i=n$, and
$\gamma_i\ge 1$ for $i\in [t]$.
Then, we let 
\[
Q_F:=\{\bp\in \R^t:\, \bp\mbox{ satisfies \eqref{eq:comp-mbb} and }\bar M\bp\le \bar\gamma\}.
\]
Let $P_F$ denote the feasible region of \eqref{LP:PF}. Using that all
right hand sides in \eqref{eq:comp-mbb} are 0, we see that 
\[
P_F\subseteq Q_F\subseteq n^2 P_F.
\]
Since $Q_F$ is an M2VPI system, Theorem~\ref{thm:2v} provides a
strongly polynomial algorithm to obtain the prices $\bp$ maximizing $\sum_{i=1}^t \bp_i$
over $Q_F$, or concludes that this objective is unbounded on $Q_F$.
In case a finite optimum exists, $\bp/n^2$ is feasible to
\eqref{LP:PF} and is within a factor $n^2$ from an optimal solution.

If the objective is unbounded on $Q_F$, then we claim that we can get a nonzero solution to \eqref{LP:PF-0}. Using LP duality, the objective is  unbounded on $Q_F$ if and only if there is a feasible solution $\bp \neq 0$ to 
\[
Q_F^0  = \{\bp \in \R^t:\, \bp\mbox{ satisfies \eqref{eq:comp-mbb} and } \bar M\bp \le 0\}.
\]
Again, Theorem~\ref{thm:2v}  is applicable to find a nonzero
  solution $q$. 
Suppose $q \neq 0$ is a solution to $Q_F^0$. This implies that $\alpha q$ is a feasible solution to $Q_F$ for all $\alpha \ge 0$. Since for every feasible solution $\bp$ to $Q_F$, $\bp/n^2$ is a feasible solution to \eqref{LP:PF}, this further implies that $\alpha q$ is also a feasible solution to \eqref{LP:PF} for all $\alpha \ge 0$. Therefore, $q$ must be a solution to \eqref{LP:PF-0}. 

The number of nonzero entries in $M$ is $\le 2n$. Thus,
constructing $\bar M$ and $\bar \gamma$ takes $O(n^4)$ time. We obtain
an M2VPI system with $\le m+2n$ constraints and $\le n$ variables, and
$m=O(n^2)$, thus  the running time for solving the M2VPI system is
$O(n^4 \log n)$ that dominates the total running time.

Finally, for
a strongly polynomial algorithm we also have to provide polynomial
bounds on the encoding lengths of the numbers during the
algorithm. The entries of $M$ are simple expressions of the input
parameters $u_{ij}$. Then, Theorem~\ref{thm:Z0} guarantees that $\bar
M$ and vector $\bar \gamma$ also have bounded encoding length. Thus, the strongly polynomial M2VPI algorithm takes a
polynomial size input and therefore the overall algorithm will be
strongly polynomial.\hfill$\Box$\medskip

\section{Approximating systems with $\Z$-matrices}\label{sec:approximate}

This section is dedicated to the proof of Theorem~\ref{thm:Z0}. 
For consistency with the market terminology, we use $\bp\in \R^t$ as the
variables. 
Recall that we need to show that
given a system $P_M=\{\bp\in \R^t:\, M\bp\le \gamma,\ \bp\ge 0 \}$ with $M\in {\cal \Z}(k,t)$ with 
$\ell$ nonzero entries and $\gamma\in \Q^k, \gamma > 0$, we can
construct a matrix $\bar M\in {\cal M}_2(\ell',t)$ and $\bar \gamma
\in \Q^{l'}$ for $\ell'\le \ell$ in $O(\ell t^3)$ time such that $P_M
\subseteq \{\bp \in \R^t:\, \bar M \bp \le \bar \gamma,\ \bp \ge 0\}
\subseteq B^2 P_M$, where $B=\sum_{j=1}^k \gamma_j/\min_{i\in [k]} \gamma_i$.

Let $M_i\in \R^t$ denote the $i$-th row of the matrix $M$ for $i\in [t]$.
We will assume that $k=t$, that is, $M$ is a square
matrix. Indeed, if $k>t$, then the last $k-t$ rows only contain
nonpositive coefficients. Therefore, for $i>t$, $M_i\bp\le 0$ holds for
every $\bp\ge 0$. If $t>k$, then by the $\Z$-property, all entries of
the last $t-k$ columns must be 0, and thus, we can delete these
columns.
We further assume that all diagonal entries are strictly positive; if
$M_{ii}=0$ then we can remove the $i$-th row and $i$-th column similarly.
Let us also note that 
every matrix in ${\cal \Z}(t,t)$ can be obtained in the form \eqref{eq:M-def},
corresponding to a market problem with $t$ components.

\paragraph{Lower bounding the vector $M\bp$}
For $i\in [t]$, we let $\lambda_i:=\sum_{j\neq i} \gamma_j$, and we let
$\lambda:=(\lambda_1,\ldots,\lambda_t)^\top$. 

\begin{lemma}\label{lem:surplus-lower}
Let $M\in {\cal \Z}(t,t)$.
Assume that $\bp\in \R^t_+$ satisfies 
$M\bp\le \gamma$.
Then, we also have
\[
M\bp\ge -\lambda\ge -(B-1) \gamma.
\]
\end{lemma}
\begin{proof}
The sum of the rows is $\sum_{j\in [t]} M_j\ge 0$ since $M$ is a $\Z$-matrix. Therefore, for any $i\in
[t]$, we have
\[
M_i \bp=\left(\sum_{j\in[t]} M_j\right) \bp - \left(\sum_{j\in [t],
    j\neq i} M_j \right)\bp \ge - \sum_{j\in [t],j\neq i} \gamma_j = -\lambda_i.
\]
The inequality $\lambda\le (B-1)\gamma$ follows by the definition of $B$.
\end{proof}

\subsection{Gaussian elimination for $\Z$-matrices} 
We will use Gaussian elimination to generate new constraints. For this
purpose, we show that Gaussian elimination on $\Z$-matrices will only add nonnegative multiples of rows to
other rows. 
\begin{lemma}\label{lem:gauss}
Let $T\in {\cal \Z}(\ell,t)$. Then, using Gaussian elimination, we can obtain a matrix
$T'=Y T$, where $T'\in \R^{\ell \times t}$ is an upper triangular matrix with diagonal
entries 0 or 1,  and all off-diagonal entries are nonpositive; further, all
entries of  $Y\in \R^{\ell \times \ell}$ are
nonnegative. If  $T_{ik}<0$ for some $k\in [t]$, $i\in [k+1,\ell]$, then $T'_{kk}=1$.
\end{lemma}
\begin{proof}
Let $T^{(k)}=Y^{(k)} T$ be the matrix after
$k$ steps in the Gaussian elimination with $T^{(0)}=T$ and $Y^{(0)}=I_\ell$.
By induction on $k$, we simultaneously show the following: 
\begin{itemize}
\item $Y^{(k)}$ is a nonnegative
matrix; 
\item $\sum_{i=k+1}^{\ell} T^{(k)}_{ij}\ge
0$ for $j\in [k+1,t]$;  
\item $T_{ij}^{(k)}\le 0$ for $i\neq j$; 
\item $T_{ii}^{(k)}\ge
0$ for all $i\in [k]$. 
\end{itemize} 
Note that the last three properties imply that the lower right $(\ell-k)\times
(t-k)$ submatrix of $T^{(k)}$ is a $\Z$-matrix.

The properties clearly hold for $k=0$; assume we have proved these for $k-1$.
 Consider the $k$-th iteration.  If $T_{kk}^{(k-1)}=0$, then no
row operation is performed. In this case, we set $T^{(k)}:=T^{(k-1)}$
and $Y^{(k)}:=Y^{(k-1)}$. We only need to verify that $\sum_{i=k+1}^{\ell} T^{(k)}_{ij}\ge
0$ for $j\in [k+1,t]$. This follows from the induction
hypotheses:  $\sum_{i=k}^{\ell} T^{(k-1)}_{ij}\ge
0$, and $T^{(k-1)}_{kj}\le 0$ for $j\in [k+1,t]$.

If $T^{(k-1)}_{kk}>0$, then we multiply the $k$-th row by
$1/T_{kk}^{(k-1)}$, and add $-T^{(k-1)}_{ik}/T^{(k-1)}_{kk}$ times the
$k$-th row to the $i$-th row for
all $i\in [k+1,\ell]$. By induction, these coefficients are all
nonnegative. We update the transformation matrix $Y^{(k)}$ accordingly, and thus it remains a
nonnegative matrix.
Consider now the $j$-th column of $T^{(k)}$ for $j\in [k+1,t]$. We have
\begin{equation}\label{eq:gauss-ind}
\sum_{i=k+1}^\ell T_{ij}^{(k)}= \sum_{i=k+1}^\ell T_{ij}^{(k-1)} -
T_{kj}^{(k-1)} \sum_{i=k+1}^\ell T_{ik}^{(k-1)} /T_{kk}^{(k-1)}.
\end{equation}
The second induction hypothesis for $j=k$ gives
\[T_{kk}^{(k-1)}+\sum_{i=k+1}^\ell T_{ik}^{(k-1)}\ge 0.\] Rearranging,
and using that $T_{kk}^{(k-1)}>0$, we obtain
\[-\sum_{i=k+1}^\ell T_{ik}^{(k-1)}/ 
T_{kk}^{(k-1)}\le 1.\]  
If we multiply this by $T_{kj}^{(k-1)} \le 0$,
we get
\[
-T_{kj}^{(k-1)}\sum_{i=k+1}^\ell T_{ik}^{(k-1)}/ 
T_{kk}^{(k-1)}\ge T_{kj}^{(k-1)}. 
\]
Substituting into \eqref{eq:gauss-ind}, this yields 
\[
\sum_{i=k+1}^\ell T_{ij}^{(k)} \ge \sum_{i=k+1}^\ell T_{ij}^{(k-1)} +T_{kj}^{(k-1)} \ge 0,
\]
again by the induction hypothesis. 

For the last part, let $T_{ik}<0$ for $k\in [t]$, $i\in [k+1,\ell]$. Note that $T_{ik}^{(k-1)}\le
T_{ik}^{(k-2)}\le\ldots\le T_{ik}^{(0)}<0$. The  induction hypothesis gives $\sum_{j=k}^{\ell} T^{(k-1)}_{jk}\ge
0$, and therefore $T_{kk}^{(k-1)}>0$. We set $T_{kk}^{(k)}=1$, and
this entry does not change in any later steps of the algorithm.
\end{proof}

\subsection{Constructing the approximate system}
Let us now describe the construction of the M2VPI system $\bar M p\le
\bar \gamma$ as in
Theorem~\ref{thm:Z0}. 
We define a digraph $([t],H)$ by adding an arc $ij \in H$ if $M_{ij}<0$.
For each $i\in [t]$, we let $D_i\subseteq [t]$ be the set of vertices
reachable from $i$ in the digraph $([t],H)$, and let $d_i:=|D_i|$. We
let $M^{(i)}$ denote the $d_i\times t$ submatrix of $M$ comprising the
rows $M_j$ for $ j\in D_i$. 
We partition $[t]$ into three groups:
\[
\begin{aligned}
T_1&:=\{i\in [t]:\, |\{j: ij\in H\}|\le 1\}, & & &\\
T_2&:=\{i\in [t]\setminus T_1:\, \rk(M^{(i)})=d_i\}, \\
T_3&:=\{i\in [t]\setminus T_1:\, \rk(M^{(i)})<d_i\}.
\end{aligned}
\]
If $i\in T_1$, then $M_i$ has at most one positive and at
most one negative entry; thus,
we can keep the constraint $M_i\bp\le \gamma_i$ unchanged. For $i\in T_2\cup T_3$,
 for every outgoing arc $ij\in H$, we shall define a constraint in the
 form ${v^{(ij)}}^\top \bp
\ge \delta^{(ij)}$. Further, for every $i\in T_2$, we shall add an
additional constraint $\bp_i\le \kappa_i$.

The construction is somewhat technical, even though the underlying
idea is relatively simple. For each $ij\in H$, we wish to obtain the constraint ${v^{(ij)}}^\top \bp
\ge \delta^{(ij)}$ such that $\bp_j$ has a positive coefficient,
$\bp_i$ has a nonpositive coefficient, and all other coefficients are
0. We wish to derive a valid constraint  for \eqref{LP:PF} by taking a
nonnegative combination of constraints from 
$M\bp\ge -\lambda$; recall from Lemma~\ref{lem:surplus-lower} that
these are valid for \eqref{LP:PF}.
Lemma~\ref{lem:gauss} shows that when we apply Gaussian elimination to
a $\Z$-matrix, then we only add rows with nonnegative
coefficients. Hence, if we apply Gaussian elimination to the matrix
$M$, and apply the same operations to the right hand side $-\lambda$,
then we can derive valid constraints from $M\bp\ge -\lambda$. In the
construction that follows, we apply a permutation to a submatrix of
$M$ where in the penultimate step of Gaussian elimination produces a
constraint of the desired form.

\medskip

Let $i\in T_2\cup T_3$, and $d:=d_i$.
For every $ij\in H$, let
us define a permutation $\sigma^{(ij)}$ of the set $[t]$ as
follows.
We set $\sigma(d-1)=j$, $\sigma(d)=i$, and fill 
the first $d-2$
positions  with the elements of $D_i\setminus\{i,j\}$ in such a way
that for any $k\in D_i\setminus\{i\}$, there is an edge
$\ell k\in H$ such that $\sigma(k)<\sigma(\ell)\le d$.
The final $t-d$ positions contain the elements of $[t]\setminus D_i$
in an arbitrary order.
Let $M^{(ij)}\in \R^{d\times t}$ denote the matrix obtained from
$M^{(i)}$ by applying the permutation  $\sigma^{(ij)}$ to the rows and the
columns, and deleting the last $t-d$ rows. 
 It is easy to see that $M^{(ij)}$ is a $\Z$-matrix.\medskip\medskip

\examplebox{
\begin{multicols}{2}
{\it Example.} The picture shows the graph $([t],H)$ for the
system obtained from Figure~\ref{fig:boost-1}. We have $D_1=D_2=D_3=\{1,2,3\}$,
and $D_4=D_5=\{1,2,3,4,5\}$. Now, 
\[
\begin{aligned}
M^{(1)}&=\begin{pmatrix*}[r]
2 & -1 & -1 & 0& 0\\
-1 & 5 &-1 & 0&0\\
-1 & -2 & 3 & 0&0
\end{pmatrix*},
\\
M^{(4)}&=\begin{pmatrix*}[r]
2 & -1 & -1 & 0&0\\
-1 & 5 &-1& 0&0\\
-1 & -2 & 3 & 0&0\\
0 & -2 & -1 & 1&-1\\
0 & 0 & 0 &-1 & 1
\end{pmatrix*},
\end{aligned}
\]
and $M^{(2)}=M^{(3)}=M^{(1)}$, $M^{(5)}=M^{(4)}$. We get $T_1=\{5\}$,
$T_2=\{1,2,3\}$, and $T_3=\{4\}$. 
\columnbreak
\phantom{x}\\
~
\begin{tikzpicture}[->,>=stealth',shorten >=1pt,auto,node distance=3cm,
                    thick,main node/.style={circle,draw}]

  \node[main node] (1) {1};
  \node[main node] (2) [below left of=1] {2};
  \node[main node] (4) [below right of=2] {4};
  \node[main node] (3) [below right of=1] {3};
  \node[main node] (5) [below of =4] {5};
  \draw    (1) edge  (3) ;
  \draw (3) edge [bend right] (1)    ;
  \draw    (2) edge  (1) ;
  \draw (1) edge [bend right] (2)    ;
 \draw    (3) edge  (2) ;
  \draw (2) edge [bend right] (3)    ;
 \draw    (4) edge [bend left]  (2) ;
  \draw (4) edge [bend right] (3)    ;
  \draw (4) edge [bend right] (5)  ;
\draw (5) edge [bend right] (4)  ;
\end{tikzpicture}
\end{multicols}
}\medskip\medskip

Let us apply Gaussian elimination as in Lemma~\ref{lem:gauss} to
$M^{(ij)}$ to obtain an upper triangular matrix $N^{(ij)}=Y^{(ij)}
M^{(ij)}$. Let $\gamma^{(ij)},\lambda^{(ij)}\in \R^d$ be the vectors obtained by permuting the components of
$\gamma$ and $\lambda$ with $\sigma^{(ij)}$, and removing the last
$t-d$ entries.

Let us set $v^{(ij)}$ to be the $(d-1)$-st row $N^{(ij)}_{d-1}$ with
the inverse of the permutation $\sigma^{(ij)}$ applied to its
elements. (So that its $i$-th coordinate corresponds to $\bp_i$).  
 Let 
\begin{equation}\label{eq:delta-def}
\delta^{(ij)}:=-Y^{(ij)}_{d-1} \lambda^{(ij)}.
\end{equation}
For $i\in T_2$, we add an additional constraint $\bp_i\le \kappa^{(i)}$. Let us pick an
arbitrary $ij\in H$, and let
\begin{equation}\label{eq:kappa-def}
\kappa^{(i)}:=Y^{(ij)}_d \gamma^{(ij)}.
\end{equation}
It will be shown in Lemma~\ref{lem:q-r} that this value is independent of the choice of the
arc $ij$. The LP $\bar M\bp \le \bar \gamma$ will be the following system:
\begin{equation}\label{LP:BM}
\begin{aligned}
M_{i}\bp&\le \gamma_i  &\forall i\in T_1,\\
\bp_i&\le \kappa^{(i)} & \forall i\in T_2,\\
{v^{(ij)}}^\top \bp &\ge \delta^{(ij)}& \forall ij\in H, i\in
T_2\cup T_3, \\
\bp&\ge 0. &
\end{aligned}
\end{equation}

\begin{figure}[!h]
\examplebox{
{\it Example.} Continuing with the  example,
we have $\gamma^\top =(2, 3, 3, 1,1)$, and $\lambda^\top=(8, 7, 7, 9, 9)$.
Let us consider $i=1$, $j=2$. We use the permutation
$\sigma^{(12)}=(3 2 1 4 5)$, yielding 
\[
M^{(12)}=\begin{pmatrix*}[r]
3 &  -2 & -1 & 0 & 0\\
-1 & 5 & -1 &  0& 0\\
- 1 & -1 & 2 & 0 & 0\\
\end{pmatrix*}, \quad 
\gamma^{(12)}= \begin{pmatrix} 3 \\ 3 \\2 \end{pmatrix}, \quad 
\lambda^{(12)}= \begin{pmatrix} 7 \\ 7 \\8 \end{pmatrix}, \quad 
\]
From Gaussian elimination, we get 
\[
\begin{aligned}
N^{(12)}&=\begin{pmatrix*}[r]
1 &  -\frac{2}3 & -\frac13 & 0 & 0\\[1ex]
0 & 1  & -\frac{4}{13} &  0& 0\\[1ex]
0 & 0 & 1& 0& 0\\
\end{pmatrix*}, &\
Y^{(12)}\gamma^{(12)}= \begin{pmatrix*}[c] 1\\[1ex]\frac{12}{13} \\[1ex] \frac{59}{15} \end{pmatrix*}, & \ \
Y^{(12)}\lambda^{(12)}= \begin{pmatrix*}[c] \frac{7}3 \\[1ex]\frac{28}{13} \\[1ex] \frac{181}{15} \end{pmatrix*}.
\end{aligned}
\]
This yields the constraint ${v^{(12)}}^\top \bp \ge \delta^{(12)}$ for
${v^{(12)}}^\top=\left(-\frac{4}{13}, 1, 0,0,\right)$, and
$\delta^{(12)}=-\frac{28}{13}$, that is, 
\[
-\frac{4}{13}\bp_1 +\bp_2\ge -\frac{28}{13}.
\]
Further, we can also use this to obtain $\bp_1\le \kappa^{(1)}$ for
$\kappa^{(1)}=\frac{59}{15}$, that is,
\[
\bp_1\le \frac{59}{15}.
\]
The system $Q_F$ comprises the constraint set
\eqref{eq:comp-mbb}, and the following constraints:
\[
\begin{aligned}
\bp_2 -\frac{4}{13}\bp_1 &\ge -\frac{28}{13}, & \bp_1&\le \frac{59}{15}, &  (i=1, j=2) & \ \ \ \  & \bp_2&\ge -\frac{88}{15}, & (i=4, j=2)\\
\bp_1-\bp_2 & \ge -\frac{31}{5}, & \bp_2&\le  \frac{32}{15}, & (i=2, j=1) & & \bp_3&\ge -\frac{154}{15}, & (i=4, j=3)\\
\bp_3 - \frac{7}{13}\bp_1 & \ge -\frac{49}{13}, & & & (i=1, j=3) & & \bp_5-\bp_4&\ge -9, & (i=4, j=5)\\
\bp_1-\frac23 \bp_3&\ge -\frac{47}9, &   \bp_3 & \le \frac{61}{12}, & (i=3, j=1) & & \bp_4-\bp_5&\ge -31. &  (i=5, j=4)\\
\bp_3- \bp_2&\ge -\frac{22}5 ,&  & & (i=2, j=3) & & \\
\bp_2-\frac13 \bp_3&\ge -\frac{22}{9}, & & & (i=3, j=2) & &.
\end{aligned}
\]
}
\end{figure}

\subsection{Proof of correctness}
We need one more claim before proving Theorem~\ref{thm:Z0}.

\begin{claim}\label{cl:diag}
Let $i\in [t]$ and let $d:=d_i$.
\begin{enumerate}[(i)] 
\item  For any $ij\in H$, $N_{kk}^{(ij)}=1$
for all $k\in [d-1]$, and $N^{(ij)}_{k\ell}=0$ for all $k\in [d]$,
$\ell\in [d+1,t]$. 
\item   If $i\in T_2$, then $N_{dd}^{(ij)}=1$, and if
$i\in T_3$, then $N_{dd}^{(ij)}=0$.
\item If $i\in T_3$, then $M_i$ can be written as a linear combination
  of the vectors $\{M_h: h\in D_i\setminus\{ i\}\}$.
\end{enumerate}
\end{claim}
\begin{proof}
{\em (i)}
In the definition of the permutation
$\sigma^{(ij)}$ it was required that for any $k\in [d-1]$, there
is an entry $M^{(ij)}_{\ell k}<0$ for some $\ell$ with
$\sigma(k)<\sigma(\ell)\le d$.
The
last part of Lemma~\ref{lem:gauss} guarantees that $N_{kk}^{(ij)}=1$
for all $k\in [d-1]$. 
Further, we note that  according to the definition of the set $D_i$, we have $M^{(ij)}_{k\ell}=0$ for all
$k\in [d]$, $\ell\in [d+1,t]$. Thus, the last $t-d$ entries of every row
in $M^{(ij)}$ are 0's and therefore also in $N^{(ij)}$. 

\medskip

\noindent
{\em (ii)}  Together with the fact the $N^{(ij)}$ is an upper triangular matrix,
we see that the only nonzero entry of the $d$-th row $N^{(ij)}_d$
is the diagonal entry $N^{(ij)}_{dd}\in \{0,1\}$. Also note that $\rk(N^{(ij)})=\rk(M^{(ij)})=\rk(M^{(i)})$.
If $i\in T_2$, then we must have 
$\rk(N^{(ij)})=d$, and therefore
$N^{(ij)}_{dd}=1$, as otherwise 
$N^{(ij)}_d=0$. On the other hand, if $i\in T_3$, then
$N^{(ij)}_{dd}=0$, as otherwise  $N^{(ij)}$ would be an upper
triangular matrix with the first $d$ diagonal entries equal to one,
contradicting $\rk(N^{(ij)})<d$.

\medskip

\noindent {\em (iii)} We claim that $\rk(M_h: h\in D_i\setminus\{i\})=d-1$. This
implies the statement, since we assumed that $\rk(M^{(i)})=\rk(M_h:
h\in D_i)=d-1$.
Note that the first $d-1$ rows of $M^{(ij)}$ are permutations of the
vectors $M_h$ for $h\in D_i\setminus\{i\}$, and hence have the same
rank. The rank is preserved during Gaussian elimination. Therefore,
\[\begin{aligned}
\rk(M_h: h\in D_i\setminus\{i\})& =\rk(M^{(ij)}_h: h\in [d-1])\\ & = \rk(N^{(ij)}_h: h\in [d-1])=d-1. 
\end{aligned}
\]
The last equality follows since $N^{(ij)}$ is an upper triangular
matrix with the first $d-1$ diagonal entries being 1, according to 
part {\em (i)}.
\end{proof}

{\it Proof of Theorem~\ref{thm:Z0}.}
\noindent {\bf Form of the constraints.}
First, let us show that the system $\bar M\bp\le \bar \gamma$ given in 
\eqref{LP:BM} is an M2VPI system. This is clearly
true 
for the constraints for $i\in T_1$ and
for $i\in T_2$.
Consider now the constraints ${v^{(ij)}}^\top \bp \ge
\delta^{(ij)}$. The vector $v^{(ij)}$ was obtained as the appropriate
permutation of the $(d-1)$-st
row of the matrix $N^{(ij)}$. Using that $N^{(ij)}$ is an upper
triangular matrix and Claim~\ref{cl:diag}(i), this
row may contain
nonzero entries only in positions $d-1$ and $d$. Further, $v^{(ij)}_{j}=
N^{(ij)}_{(d-1)(d-1)}=1$, and 
$v^{(ij)}_{i}=N^{(ij)}_{(d-1)d}\le 0$ as it is an off-diagonal entry.

\medskip

\noindent {\bf Containment of $P_M$.}
We show that every $\bp$ satisfying $M\bp\le\gamma$ is also 
 feasible to \eqref{LP:BM}. For $i\in T_1$,
the constraint
$M_i \bp \le \gamma_i$ is identical to the $i$-th constraint in $P_M$. 

For $i\in T_2$, let $ij\in H$ be the edge used in the definition of
$\kappa^{(i)}$.  According to Claim~\ref{cl:diag},  the row $N^{(ij)}_{d}$ has a single
nonzero entry $N^{(ij)}_{dd}=1$.
Lemma~\ref{lem:gauss} guarantees that the coefficient matrix $Y^{(ij)}$ is nonnegative.
Therefore, the constraint $\bp_i\le
\kappa^{(i)}$ can be obtained as a nonnegative combination of the constraint set
$M\bar p\le \gamma$, by multiplying $M_h
\bp\le \gamma_h$ for $h\in D_i$ by $Y^{(ij)}_{d\sigma^{(ij)}(h)}$. 

The validity of the constraints ${v^{(ij)}}^\top \bp \ge
\delta^{(ij)}$ follows similarly. Recall from Lemma~\ref{lem:surplus-lower} that
$M\bp\ge -\lambda$ is valid for $\bp$.  The constraint ${v^{(ij)}}^\top \bp \ge
\delta^{(ij)}$ is obtained by taking a
nonnegative combination of the inequalities $M\bp\ge -\lambda$
combining $M_h\bp\ge -\lambda_h$ for $h\in D_i$ with the nonnegative coefficient $Y^{(ij)}_{(d-1)\sigma^{(ij)}(h)}$.
Hence, all these inequalities are valid for $\bp$.

\medskip

\noindent {\bf Approximate reverse containment.}
We next show that if $\bp$ is feasible to \eqref{LP:BM}, then $\bp$ is
feasible to $B^2 P_M$, that is, $M\bp\le B^2\gamma$.
Clearly, for $i\in T_1$,  $M_i\bp\le \gamma_i\le B^2\gamma_i$. The
more difficult part is to show the validity of $M_i\bp\le B^2\gamma_i$ for
$i\in T_2\cup T_3$.

For $i\in T_3$, we show that the constraints 
\[
{v^{(ij)}}^\top \bp \ge
\delta^{(ij)}\quad \forall j:\, ij\in H
\]
together imply $M_i \bp\le B^2\gamma_i$. For $i\in T_2$, we will also make use of the additional
constraint $\bp_i\le \kappa^{(i)}$ to derive  $M_i \bp\le
B^2\gamma_i$.
The following technical lemma will be needed.

\begin{lemma}\label{lem:q-r}
Consider any $i\in T_2\cup T_3$.
\begin{enumerate}[(i)]
\item There is a unique vector $q^{(i)}\in \R^t_+$ such that $M_\ell
  q^{(i)}=0$ for all $\ell\in D_i\setminus\{i\}$, $q_i^{(i)}=1$,
  and $q_\ell^{(i)}=0$
  for $\ell\in [t]\setminus D_i$.
\item For any $ij\in H$, $v^{(ij)}_j=1$ and $v^{(ij)}_i=-q^{(i)}_j$.
\item If $i\in T_2$, then there exists a vector $r^{(i)}\in \R^t_+$ with
  $Mr^{(i)}\le \gamma$, $M_\ell
  r^{(i)}=\gamma_\ell$ for all $\ell\in D_i$, and
  $r_i^{(i)}=\kappa^{(i)}$.
\item If $i\in T_3$, then there exists  a vector $r^{(i)}\in \R^t_+$
  with $Mr^{(i)}\le \gamma$, and $M_\ell
  r^{(i)}=\gamma_\ell$ for all $\ell\in D_i\setminus \{i\}$. 
\end{enumerate}
\end{lemma}
\begin{proof}
As before, we let $d=d_i$. For part {\em (i)}, let us apply Gaussian elimination as in
Lemma~\ref{lem:gauss} to the matrix $M^{(ij)}$ for any $ij\in H$. 
Note that $M^{(ij)}_\ell \hat q=0$ for $\ell \in [d-1]$ is equivalent to $M_\ell q=0$ for
all $\ell \in D_i\setminus \{i\}$, where $\hat q$ is the vector
obtained from $q$ by permuting the elements by $\sigma^{(ij)}$.
The system of  linear equations $M^{(ij)}_\ell\hat q=0$ for $\ell \in [d-1]$ is turned into
the system $N^{(ij)}_\ell \hat q=0$ for $\ell \in [d-1]$ using
the elimination.

From Claim~\ref{cl:diag}, we have that $N_{kk}^{(ij)}=1$ for all $k\in
[d-1]$, and all
off-diagonal elements of the upper triangular matrix $N^{(ij)}$ are
nonpositive.
Hence, if we
set $\hat q_k=0$ for all $k\in [d+1,t]$ and $\hat q_{d}=1$, then the system
$N^{(ij)}_\ell \hat q=0$  for $\ell\in [d-1]$ has a unique nonnegative solution. Applying the
inverse of $\sigma^{(ij)}$ to $\hat q$ we obtain the desired vector
$q^{(i)}$.

\medskip

For part {\em (ii)}, recall that $v_i^{(ij)}$ is obtained from
$N_{d-1}^{(ij)}$ using the inverse of the permutation $\sigma^{(ij)}$. We have
already verified that the only possible nonzero entries in
$N_{d-1}^{(ij)}$ are the $(d-1)$-st and $d$-th components, and that 
$v_j^{(ij)}= N_{(d-1)(d-1)}^{(ij)}=1$. Since $N_{d-1}^{(ij)}\hat q=0$, and $\hat
q_{d}=q^{(i)}_i=1$, we must have $\hat q_{d-1}
+N_{(d-1)d}^{(ij)}=0$, and thus $N_{(d-1)d}^{(ij)} =-\hat q_{d-1}$; after
permuting, this gives $v_i^{(ij)}=-q_j^{(i)}$.

\medskip

Let us now turn to part {\em (iii)}. Let $i\in T_2$, and let us fix an arbitrary
$ij\in H$. The system $M_\ell r=\gamma_\ell$ for all
$\ell\in D_i$ is equivalent to $M^{(ij)} \hat
r=\gamma^{(ij)}$, where $\hat r$ is
obtained from $r$ by the permutation $\sigma^{(ij)}$. After Gaussian
elimination, we obtain the equivalent system 
\[N^{(ij)}\hat r=Y^{(ij)} \gamma^{(ij)}.\]
According to Claim~\ref{cl:diag}, $N^{(ij)}_{\ell \ell}=1$ for all $\ell\in
[d]$. This system has a unique solution $\hat r$ with
$\hat r_\ell =0$ for $d+1\le \ell\le t$; note also that $\hat
r_d=Y^{(ij)}_d \gamma^{(ij)}=\kappa^{(i)}$. Since the right hand side
is nonnegative, and $N^{(ij)}$ is an upper triangular matrix with all
off-diagonal entries being nonpositive, we see that $\hat r_\ell\ge 0$
for all $\ell\in [d]$.
 Let $r^{(i)}$ be the
vector obtained from $\hat r$ by  applying the inverse
permutation of $\sigma^{(ij)}$. It remains to show that $M_\ell r^{(i)}\le
\gamma_\ell$ holds for all $\ell\in [t]\setminus D_i$. The only
positive coefficient in $M_\ell$ can be $M_{\ell\ell}$. However, since
$r^{(i)}_\ell=0$, it follows that $M_\ell r^{(i)}\le 0$ for these
values of $\ell$.

\medskip

The proof of part {\em (iv)} is similar. Let $i\in T_3$, and let us
fix an arbitrary $ij\in H$.  The system $M_\ell r=\gamma_\ell$ for all
$\ell\in D_i\setminus \{i\}$ is equivalent to $M^{(ij)}_\ell \hat
r=\gamma^{(ij)}_\ell$ for $\ell \in [d-1]$, where $\hat r$ is obtained
by applying the 
permutation $\sigma^{(ij)}$ to the vector $r$. After Gaussian
elimination, we obtain the equivalent system
\[
N^{(ij)}_\ell \hat r=Y^{(ij)}_\ell \gamma^{(ij)}\quad \forall \ell\in [d-1],
\]
with $N^{(ij)}_{\ell
  \ell}=1$ for all $\ell\in [d-1]$. As in part (iii), we see that
there is a unique  solution with $\hat r_\ell=0$ for $d\le \ell\le t$,
and this solution satisfies $\hat r_\ell\ge 0$ for all $\ell\in [d-1]$. 
We obtain $r^{(i)}$ by applying the inverse of $\sigma^{(ij)}$ to
$\hat r$. It follows as above that  $M_\ell r^{(i)}\le
0$ for all $\ell\in [t]\setminus D_i$. The same argument also
gives $M_i r^{(i)}\le
0$, since $r^{(i)}_i=0$. Hence, $Mr^{(i)}\le \gamma$ holds.
\end{proof}

Assume now that $i\in T_2\cup T_3$. We  show that $M_i \bp\le
B^2\gamma_i$. 
 Let $q:=q^{(i)}$ as in Lemma~\ref{lem:q-r}(i). By part (ii)
of the same lemma, and substituting the definition
\eqref{eq:delta-def} of $\delta^{(ij)}$,
 the constraints can be written as
\[
\bp_j - q_j \bp_i \ge -Y_{d-1}^{(ij)}\lambda^{(ij)}  \quad \forall j:
ij\in H.
\]
Note that $\lambda^{(ij)}_\ell \le (B-1)\gamma^{(ij)}_\ell$ for all
$\ell\in [d]$ by the definition of $B$, and
$Y_{d-1}^{(ij)}\ge 0$. Therefore, these constraints
imply 
\[
\bp_j - q_j \bp_i \ge -(B-1)Y_{d-1}^{(ij)}\gamma^{(ij)}  \quad \forall j:
ij\in H.
\]
Recall that $M_{ij}<0$ if and only if $ij\in H$.
Let us multiply the inequality for every $j\neq i$ by $M_{ij}\le 0$, and add
up these inequalities. We obtain
\begin{equation}\label{eq:p-r-comb}
\sum_{j: ij\in H}  M_{ij}\bp_j - \left(\sum_{j: ij\in H} M_{ij}q_j
\right)\bp_i \le -(B-1) \sum_{j: ij\in H} M_{ij} Y_{d-1}^{(ij)}\gamma^{(ij)} .
\end{equation}
For the rest of the proof, we distinguish the cases $i\in T_3$ and  $i\in T_2$.

\medskip

\noindent
{\bf Case $i\in T_3$.} Since $M_h q=0$ for all $h\in D_i\setminus
\{i\}$, Claim~\ref{cl:diag}(iii) implies $M_i q=0$.
Substituting $q_i=1$, we see that $M_{ii}=- \sum_{j: ij\in H} M_{ij}q_j$.
With $\eta_i:=-\sum_{j: ij\in H} M_{ij} Y_{d-1}^{(ij)}\gamma^{(ij)} $,
\eqref{eq:p-r-comb} can be written as
\begin{equation}\label{eq:p-r-2}
M_{ii}\bp_i +\sum_{j: ij\in H}  M_{ij}\bp_j \le (B-1)\eta_i.
\end{equation}
The left hand side is $M_i \bp$. We next show $\eta_i\le \lambda_i$,
which together with 
$\lambda_i\le (B-1) \gamma_i$ yields
$M_i \bp \le (B-1)^2\gamma_i$.

To see $\eta_i\le \lambda_i$, we make use of the vector $r=r^{(i)}$ as in
Lemma~\ref{lem:q-r}(iv).  Let $\hat r$ denote the permutation
$\sigma^{(ij)}$ applied to $r$. Since $M^{(ij)}_\ell \hat r=\gamma_\ell$ is valid
for all $\ell\in [d-1]$, we have $N^{(ij)}_{d-1} \hat
r=Y_{d-1}^{(ij)}\gamma^{(ij)}$. This can be written as
\begin{equation}\label{eq:rij}
r_j - q_j r_i= Y_{d-1}^{(ij)}\gamma^{(ij)}  \quad \forall ij\in H.
\end{equation}
Summing up these equalities after multiplying the $j$-th one by
$M_{ij}<0$, we see as above
that 
\[
M_i r = -\eta_i.
\]
Since $Mr\le \gamma$, from Lemma~\ref{lem:surplus-lower} we have
$-\eta_i = M_i r\ge -\lambda_i$, and therefore $\eta_i\le\lambda_i$ as
needed.

\medskip

\noindent
{\bf Case $i\in T_2$.} The coefficient of $\bp_i$ in
\eqref{eq:p-r-comb} equals $M_{ii}-M_iq$. In contrast with the previous case, $M_i q$ is
not necessarily $0$. We claim that $M_iq\ge 0$. To see this, note that $\sum_{h\in
  D_i} M_hq\ge 0$ since  $\sum_{h\in D_i} M_h\ge 0$ from the $\Z$-property and $q\ge 0$; further, $M_h q=0$ for $h\in D_i\setminus\{i\}$.
Let us further add to \eqref{eq:p-r-comb}
$M_i q$ times the inequality $\bp_i\le \kappa^{(i)}$. Thus, we obtain
\begin{equation}\label{eq:t-1}
M_i \bp\le (B-1)\eta_i+\kappa^{(i)}M_iq.
\end{equation}
Let $r=r^{(i)}$ as in Lemma~\eqref{lem:q-r}(iii). As for $i\in T_3$,
the equations \eqref{eq:rij} hold.
Adding up these equations multiplied by $M_{ij}$, and further adding
$M_i q$ times the equality $r_i=\kappa^{(i)}$, we obtain
\[
M_ir = -\eta_i + \kappa^{(i)}M_iq.
\]
On the other hand, we know that $M_ir=\gamma_i$. Thus,
$\gamma_i=-\eta_i+\kappa^{(i)}M_iq$, yielding $\eta_i\le \kappa^{(i)}M_iq$.
Consequently, from \eqref{eq:t-1} we obtain
\begin{equation}\label{eq:t-1-2}
M_i\bp\le B\kappa^{(i)}M_iq.
\end{equation}
The next claim completes the proof of $M_i\bp\le B^2 \gamma_i$.
\begin{claim}
$\kappa^{(i)}M_iq\le  B\gamma_i$. 
\end{claim}
\begin{proof}
For $ij\in H$ used in defining $\kappa^{(i)}$, let $\hat q$ be the
vector obtained from $q$ by applying the permutation $\sigma^{(ij)}$. We first show that
$M_iq=1/Y_{dd}^{(ij)}$.
To see this, first note that $N^{(ij)}_d\hat q=1$, since $N^{(ij)}_{dd}=1$,
$\hat q_d=1$, 
and all other entries are 0.
We have 
\[
1=N^{(ij)}_d \hat q=Y^{(ij)}_d (M^{(ij)}\hat q)=Y^{(ij)}_{dd}
(M^{(ij)}_d\hat q)= Y^{(ij)}_{dd} (M_i q)
\]
The third  equality follows since $M^{(ij)}_\ell\hat q=0$ for $\ell\in
[d-1]$. Using the definition \eqref{eq:kappa-def} of $\kappa^{(i)}$, we have that  
\begin{equation}\label{eq:kappa}
\kappa^{(i)}M_i q= \sum_{h=1}^d \frac{Y^{(ij)}_{dh}}{Y_{dd}^{(ij)}} \gamma^{(ij)}_h.
\end{equation}
Let us now show that $Y^{(ij)}_{dh}\le Y_{dd}^{(ij)}$ for
all $h\in [d]$. Let us pick $h$ such that it
maximizes $Y^{(ij)}_{dh}$ over $h\in [d]$, and select the largest $h$
where the maximum is taken. For a contradiction, assume that $h<d$.
 Since
$N^{(ij)}$ is an upper triangular matrix, we have
\[
0=N_{dh}^{(ij)}=\sum_{\ell=1}^d Y_{d\ell}^{(ij)} M^{(ij)}_{\ell h}.
\] 
Using that $M^{(ij)}$ is a $\Z$-matrix, we have $M^{(ij)}_{hh}\ge 0$,
$M^{(ij)}_{\ell h}\le 0$ for $\ell\neq h$, and $\sum_{\ell=1}^d
M^{(ij)}_{\ell h}\ge 0$. Using the maximality of $Y^{(ij)}_{dh}$, we
see that equality can only hold if  $\sum_{\ell=1}^d
M^{(ij)}_{\ell h}= 0$, and $Y_{d\ell}^{(ij)}=Y^{(ij)}_{dh}$ whenever $M^{(ij)}_{\ell h}<0$.
This contradicts the maximal choice of $h$, since the permutation
$\sigma^{(ij)}$ was chosen so that there exists an $\ell>h$ with
$M^{(ij)}_{\ell h}<0$ whenever $h<d$. Thus, we can conclude that
$Y^{(ij)}_{dh}\le Y_{dd}^{(ij)}$ for all $h\in [d]$.
From \eqref{eq:kappa}, it follows that
\[
\kappa^{(i)}M_i q\le \sum_{h=1}^d \gamma^{(ij)}_h=\sum_{h\in D_i}
\gamma_h\le B \gamma_i,
\]
using the definition of $B$.
\end{proof}

\medskip

\noindent {\bf Running time and encoding length.} 
We can obtain the constraints in \eqref{LP:BM} by running Gaussian
elimination for each $ij\in H$; this takes altogether $O(\ell (\min\{k,t\})^3)$ time.
We need to show that the encoding size of $\bar
M$ and $\bar b$ are polynomially bounded in the encoding size of $M$
and $b$. This easily follows since the constraints are obtained by
Gaussian elimination; we refer to \cite{Edmonds67} for strong
polynomiality of Gaussian elimination.

The proof of Theorem~\ref{thm:Z0} is now complete.\hfill$\Box$\medskip

\section{Conclusions}\label{sec:conclude}
We have given a strongly polynomial algorithm for computing an exact 
equilibrium in linear exchange markets.
We use a variant of  the Duan-Mehlhorn
algorithm as a subroutine in a framework that repeatedly identifies
revealed arcs. Before each iteration of this subroutine, we use
another method to find a good starting solution for the current set of
revealed arcs.
The best solution here corresponds to the optimal solution of a linear
program. Whereas no strongly polynomial algorithm is known for an LP
of this form, we presented a strongly polynomial approximation by
constructing a second LP.

The goal of this paper was to get a better
understanding of theoretical solvability of this important practical problem.
For practical purposes, the approximate algorithms mentioned in the introduction may suffice; in particular, the simple and efficient $O(\frac{n}{\varepsilon}(m+n\log n))$ algorithm by 
\cite{GhiyasvandO12} would be a natural choice.  Although 
we have not implemented our algorithm, we believe that it should be much better in practice than the worst-case estimate $O(n^{10}\log^2 n) $. We are not aware of previous computational studies on exchange market algorithms; we leave the comparison of our algorithm to \cite{Ye08} and \cite{DuanGM16}  for future work.

\medskip

From a theoretical perspective, it could be worth exploring whether this approach extends
further. An immediate question is to see if one can use such an approach to obtain a
$\varepsilon$-approximation of the LP in strongly polynomial time for
every $\varepsilon>0$. Further, such a method could be potentially
useful for a broader class of LPs; a natural candidate would be
systems of the form $A^\top x\le c$ for a pre-Leontief matrix $A$ (see \cite{Cottle72}), a class where a pointwise
maximal solution exists, but no strongly polynomial algorithm is
known.

Our approach was specific to the market equilibrium problem. The
method of identifying revealed arc sets originates from
\cite{Vegh16}. This result was applicable not only for the linear
Fisher market model, but more generally, for minimum-cost flow problems with separable
convex objectives satisfying certain assumptions. It would be
desirable to extend the current approach to a broader class of convex
programs that include the formulation in \cite{DevanurGV16}.

\subsection*{Acknowledgements}
The authors are grateful to Kurt Mehlhorn, Vijay Vazirani, and  Richard Cole for many interesting discussions on this problem.

\appendix

\section{Proof of Lemma~\ref{lem:progress}}\label{sec:dm-mp}
We will use the following simple lemmas to prove the result.

\begin{lemma}[\cite{DuanGM16}]\label{lem:bfi}
Let $r=(r_1,\dots,r_n)$ and $r'=(r_1',\dots,r_n')$ be nonnegative vectors. Let $k\in[1,n]$ be such that $r_i'\ge r_j'$ for $i\le k<j$. Suppose that $\delta_i = r_i-r_i'\ge 0$ for $i\le k$ and $\delta_j = r_j'-r_j \ge 0$ for $j>k$. Let $D=\min_{i\le k}r_i -\max_{j>k}r_j$, and let $\Delta=\sum_{i\le k}\delta_i$. If $\Delta\ge \sum_{j>k}\delta_j$, 
\[\|r'\|_{_2}^2 \le \|r\|_{_2}^2 - D\Delta\enspace .\]
\end{lemma}

\begin{lemma}[\cite{DuanGM16,DuanM15}]\label{lem:setS}
Given $n$ numbers $a_1\ge a_2\ge \dots \ge a_n\ge 0$. Let $l$ be minimal such that $a_l/a_{l+1}\le 1+1/n$. Let $l=n$ if there is no such $l$. Then $(e\cdot n)a_l \ge \|a\|_{_1}$ and $a_i \le e \cdot a_l, \forall i$. 
\end{lemma}

The proof of the following lemma is an adaptation of the proof given in~\cite{DuanGM16}.

\normProgress*
\begin{proof}
Let us denote $C=56e^2$, and let $x_{\max} = 1+\frac{1}{Cn^3}$.
Recall that price-rise phases correspond to Event 3. Since no good leaves $\Gamma(S)$ during a phase and prices monotonically increase, Event 3 implies that price of at least one good has increased by $x_{\max}$. This proves the first part. 

For the second part, the phase terminates either due to Event 2 or breaking from the inner loop (in line \ref{alg:break}). 
For the latter case, there is a way to adjust the flow so that the surplus of an agent $i\in S$ becomes equal to either zero if $S=A$ or the surplus of an agent $i'\notin S$. Furthermore, the flow adjustment maintains that the surplus of each agent $i\in S$ increases by a factor of at most $x_{max}$, surpluses of agents outside $S$ does not decrease, and the total surplus does not increase. This is a similar situation like Event 2, so we can prove this case by the proof of Event 2 case. 

Let $f''$ be the intermediate flow just before computing a balanced flow in line \ref{alg:bf}. If Event 2 occurs, then it implies that $p'_j \le x_{\max}p_j, \forall j$. Since $p_j$ and $f_{ij}$ are increased by the same factor $x, \forall i\in S, \forall j\in \Gamma(S)$, and the flow update in line \ref{alg:update} can only decrease the surplus of agents in $S$, we have $c_i(p',f'') \le x_{\max}c_i(p,f), \forall i\in S$. 
There are two cases: either the surplus of an agent $i\in S$ becomes zero if $S=A$ or it becomes equal to surplus of an agent $i'\notin S$.

For the first case, we have $c_i(p',f'') = 0$ 
and $c_j(p',f'') \le x_{\max}c_j(p,f), \forall j\in S$. This implies that 
\begin{equation}
\begin{aligned}
\|c(p',f'')\|_{_2}^2 & \le & x_{\max}^2\|c(p,f)\|_{_2}^2 - c_i(p,f)^2 \enspace . \nonumber
\end{aligned}
\end{equation}
Lemma~\ref{lem:setS} implies that $c_i(p,f) \ge \|c(p,f)\|_{_1}/(e\cdot n)$. Using this, we get 
\begin{equation}\nonumber
\begin{aligned}
\|c(p',f'')\|^2_{_2} & \le & \|c(p,f)\|^2_{_2} \displaystyle\left(\left(1+\frac{1}{Cn^3}\right)^2 - \frac{1}{e^2 n^2}\right)\enspace . 
\end{aligned}
\end{equation}
Simplifying above using $C=56e^2$, we get 
\begin{equation}\nonumber
\begin{aligned}
\|c(p',f'')\|^2_{_2} & \le & \displaystyle\left(1+\frac{3}{Cn^3} - \frac{56}{C n^3}\right)\|c(p,f)\|^2_{_2} \\ & <& \displaystyle\left(1-\frac{4}{Cn^3}\right)\|c(p,f)\|^2_{_2} \enspace . 
\end{aligned}
\end{equation}
Further, we obtain 
\begin{equation}\nonumber
\begin{aligned}
\|c(p',f'')\|_{_2} \le \sqrt{\left(1-\frac{4}{Cn^3}\right)}\|c(p,f)\|_{_2} & \le  \left(1 - \frac{2}{Cn^3}\right)\|c(p,f)\|_{_2}\\ & \le {\|c(p,f)\|_{_2}}/{\left(1+\frac{1}{Cn^3}\right)}\enspace . 
\end{aligned}
\end{equation}
Since $\|c(p',f')\|_{_2} \le \|c(p',f'')\|_{_2}$, we get \[\|c(p',f')\|_{_2} \le \|c(p,f)\|_{_2}/\left(1+\frac{1}{Cn^3}\right).\] 

For the second case, $\min_{i\in S} c_i(p', f'')=\max_{i\not\in S} c_i(p', f'')$. There are two types of agents in $S$: $i$ such that $g_i \in \Gamma(S)$ (type 1) and $i$ such that $g_i\notin \Gamma(S)$ (type 2). The price and flow update in line \ref{alg:pfupdate} increases the surpluses of type 1 agents and decreases the surpluses of type 2 agents. Since the price and flow are updated by the same factor $x$, the surpluses of type 1 agents increase by the same factor $x$. 

Since the surpluses of agents outside $S$ does not decrease, and the total surplus monotonically decreases, the total increase in the surpluses of agents outside $S$ is at most the total decrease in the surpluses of type 2 agents. For simplicity, let $w_1, \dots, w_{\ell}$ and $u_1, \dots, u_k$ denote the surpluses of type 1 and type 2 agents at $(p,f)$ respectively. Similarly, let $v_1,\dots, v_{o}$ denote the surpluses of agents outside $S$ at $(p,f)$. Clearly, we have $l+k+o=n$. Define $\bar w=\min_i w_i, \bar u=\min_i u_i, \bar v=\max_j v_j$. Let $R=\min\{\bar u,\bar w\}$. From the definition of $S$, we have $R > (1+1/n)\bar v$. At $(p',f'')$, let $u_1', \dots, u_k', w_1',\dots, w_l', v_1', \dots, v_o'$ are the surpluses of type 1, type 2, and agents outside $S$ respectively. From the above discussion, we have
\[w_i' \le x_{\max} w_i, \forall i,\ \ \ \ \text{and} \ \ \ \ u_i' \le u_i, \forall i \ \ \ \ \text{and} \ \ \ \ v_i' \ge v_i, \forall i\enspace .\]

Let $u_i' = u_i-\delta_i, \forall i$ and $v_j'=v_j+\delta_j',\forall j$, where $\delta_i, \delta_{j}' \ge 0, \forall i, j$. Further, we have $w_i'\ge v_i', \forall i$, $u_i'\ge v_i',\forall i$, and $\sum_i \delta_i \ge \sum_j \delta_j'$. Since $R-\bar v\ge R/(n+1)$ and Event 2 occurred, we have $\sum_i \delta_i \ge R/(2(n+1))$. Using Lemma~\ref{lem:bfi}, we get 
\[\|u'\|_{_2}^2 + \|v'\|_{_2}^2 \le \|u\|_{_2}^2 + \|v\|_{_2}^2 - \frac{R^2}{2(n+1)^2}\enspace . \]

Recall that $w_i'\le x_{\max}w_i, \forall i$, and let us use the trivial upper bound
$\|w\|_{_2}^2 \le \|c(p,f)\|_{_2}^2$.
Together with the above inequality, we obtain 
\begin{equation}\nonumber
\begin{aligned}
\|c(p',f'')\|_{_2}^2 & =  \|w'\|_{_2}^2 + \|u'\|_{_2}^2 + \|v'\|_{_2}^2 \\ & \le x_{max}^2\|w\|_{_2}^2 + \|u\|_{_2}^2 + \|v\|_{_2}^2 - \frac{R^2}{2(n+1)^2}\\
& \le \left(1+\frac{1}{Cn^3}\right)^2\|c(p,f)\|_{_2}^2 - \frac{R^2}{2(n+1)^2}\enspace .
\end{aligned}
\end{equation}
From Lemma~\ref{lem:setS}, we have $\|c(p,f)\|^2_{_2} \le ne^2R^2$. Hence,
\begin{equation}\nonumber
\begin{aligned}
\|c(p',f'')\|_{_2}^2 & \le  \|c(p,f)\|_{_2}^2  \left( \left(1+\frac{1}{Cn^3}\right)^2 - \frac{1}{2e^2n(n+1)^2}\right)\\
& \le {\|c(p,f)\|^2}/{\left(1+\frac{4}{Cn^3}\right)},\\  
\end{aligned}
\end{equation}
where the last inequality used that $C=56e^2$. We obtain
\begin{equation}\nonumber
\|c(p',f'')\|_{_2} \le  \frac{\|c(p,f)\|_{_2}}{\sqrt{1+\frac{4}{Cn^3}}}  \le \frac{\|c(p,f)\|_{_2}}{1+\frac{1}{Cn^3}} \enspace .
\end{equation}
Finally, since $\|c(p',f')\|_{_2} \le \|c(p',f'')\|_{_2}$, we get $\|c(p',f')\|_{_2} \le{\|c(p,f)\|_{_2}}/{\left(1+\frac{1}{Cn^3}\right)}$. 
\end{proof}

\bibliographystyle{abbrv} 
\bibliography{arrow-debreu} 

\begin{thebibliography}{10}

\bibitem{Adler91}
I.~Adler and S.~Cosares.
\newblock A strongly polynomial algorithm for a special class of linear
  programs.
\newblock {\em Operations Research}, 39(6):955--960, 1991.

\bibitem{arrow}
K.~J. Arrow and G.~Debreu.
\newblock Existence of an equilibrium for a competitive economy.
\newblock {\em Econometrica: Journal of the Econometric Society}, pages
  265--290, 1954.

\bibitem{BrainardS00}
W.~Brainard and H.~Scarf.
\newblock How to compute equilibrium prices in 1891.
\newblock {\em Cowles Foundation Discussion Paper}, 1270, 2000.

\bibitem{codenotti04}
B.~Codenotti, S.~Pemmaraju, and K.~Varadarajan.
\newblock The computation of market equilibria.
\newblock {\em ACM SIGACT News}, 35(4):23--37, 2004.

\bibitem{CohenM94}
E.~Cohen and N.~Megiddo.
\newblock Improved algorithms for linear inequalities with two variables per
  inequality.
\newblock {\em SIAM Journal on Computing}, 23(6):1313--1347, 1994.

\bibitem{cornet89}
B.~Cornet.
\newblock Linear exchange economies.
\newblock Technical report, Cahier Eco-Math, Universit{\'e} de Paris, 1989.

\bibitem{Cottle72}
R.~W. Cottle and A.~F. Veinott.
\newblock Polyhedral sets having a least element.
\newblock {\em Mathematical Programming}, 3(1):238--249, 1972.

\bibitem{DarwishM16}
O.~Darwish and K.~Mehlhorn.
\newblock Improved balanced flow computation using parametric flow.
\newblock {\em Inf. Process. Lett.}, 116(9):560--563, 2016.

\bibitem{DevanurGV16}
N.~R. Devanur, J.~Garg, and L.~A. V{\'e}gh.
\newblock A rational convex program for linear {Arrow-Debreu} markets.
\newblock {\em ACM Transactions on Economics and Computation (TEAC)}, 5(1):6,
  2016.

\bibitem{DevanurPSV08}
N.~R. Devanur, C.~H. Papadimitriou, A.~Saberi, and V.~V. Vazirani.
\newblock Market equilibrium via a primal--dual algorithm for a convex program.
\newblock {\em Journal of the ACM (JACM)}, 55(5):22, 2008.

\bibitem{DevanurV03}
N.~R. Devanur and V.~V. Vazirani.
\newblock An improved approximation scheme for computing {Arrow-Debreu} prices
  for the linear case.
\newblock In {\em Proceedings of FSTTCS}, pages 149--155, 2003.

\bibitem{DuanGM16}
R.~Duan, J.~Garg, and K.~Mehlhorn.
\newblock An improved combinatorial polynomial algorithm for the linear
  {A}rrow-{D}ebreu market.
\newblock In {\em Proc.\ 27th Symp.\ Discrete Algorithms (SODA)}, pages
  90--106, 2016.

\bibitem{DuanM15}
R.~Duan and K.~Mehlhorn.
\newblock A combinatorial polynomial algorithm for the linear {Arrow--Debreu}
  market.
\newblock {\em Information and Computation}, 243:112--132, 2015.

\bibitem{Eaves76}
B.~C. Eaves.
\newblock A finite algorithm for the linear exchange model.
\newblock {\em Journal of Mathematical Economics}, 3:197--203, 1976.

\bibitem{Edmonds67}
J.~Edmonds.
\newblock Systems of distinct representatives and linear algebra.
\newblock {\em Journal of Research of the National Bureau of Standards B},
  71:241--245, 1967.

\bibitem{EG}
E.~Eisenberg and D.~Gale.
\newblock Consensus of subjective probabilities: The pari-mutuel method.
\newblock {\em The Annals of Mathematical Statistics}, 30(1):165--168, 1959.

\bibitem{Gale76}
D.~Gale.
\newblock The linear exchange model.
\newblock {\em Journal of Mathematical Economics}, 3(2):205--209, 1976.

\bibitem{GargMVY17}
J.~Garg, R.~Mehta, V.~V. Vazirani, and S.~Yazdanbod.
\newblock Settling the complexity of {L}eontief and {PLC} exchange markets
  under exact and approximate equilibria.
\newblock In {\em Proc.\ 49th Symp.\ Theory of Computing (STOC)}, pages
  890--901, 2017.

\bibitem{GargK06}
R.~Garg and S.~Kapoor.
\newblock Auction algorithms for market equilibrium.
\newblock {\em Mathematics of Operations Research}, 31(4):714--729, 2006.

\bibitem{GhiyasvandO12}
M.~Ghiyasvand and J.~B. Orlin.
\newblock A simple approximation algorithm for computing {Arrow-Debreu} prices.
\newblock {\em Operations Research}, 60(5):1245--1248, 2012.

\bibitem{GoelV11}
G.~Goel and V.~Vazirani.
\newblock A perfect price discrimination market model with production, and a
  rational convex program for it.
\newblock {\em Mathematics of Operations Research}, 36(4):762--782, 2011.

\bibitem{Goldberg89}
A.~V. Goldberg and R.~E. Tarjan.
\newblock Finding minimum-cost circulations by canceling negative cycles.
\newblock {\em Journal of the ACM (JACM)}, 36(4):873--886, 1989.

\bibitem{glsbook}
M.~Gr{\"o}tschel, L.~Lov{\'a}sz, and A.~Schrijver.
\newblock {\em Geometric algorithms and combinatorial optimization}.
\newblock Springer Verlag, 1988.

\bibitem{HochbaumN94}
D.~S. Hochbaum and J.~Naor.
\newblock Simple and fast algorithms for linear and integer programs with two
  variables per inequality.
\newblock {\em SIAM Journal on Computing}, 23(6):1179--1192, 1994.

\bibitem{Jain07}
K.~Jain.
\newblock A polynomial time algorithm for computing an {A}rrow-{D}ebreu market
  equilibrium for linear utilities.
\newblock {\em SIAM Journal on Computing}, 37(1):303--318, 2007.

\bibitem{JainMS03}
K.~Jain, M.~Mahdian, and A.~Saberi.
\newblock Approximating market equilibria.
\newblock In {\em Proceedings of APPROX-RANDOM}, pages 98--108, 2003.

\bibitem{Kamiyama19}
N.~Kamiyama.
\newblock A note on balanced flows in equality networks.
\newblock {\em Inf. Process. Lett.}, 145:74--76, 2019.

\bibitem{Lemke}
C.~E. Lemke.
\newblock Bimatrix equilibrium points and mathematical programming.
\newblock {\em Management Science}, 11(7):681--689, 1965.

\bibitem{Megiddo83}
N.~Megiddo.
\newblock Towards a genuinely polynomial algorithm for linear programming.
\newblock {\em SIAM Journal on Computing}, 12(2):347--353, 1983.

\bibitem{NenakovP83}
E.~I. Nenakov and M.~E. Primak.
\newblock One algorithm for finding solutions of the {Arrow-Debreu} model.
\newblock {\em Kibernetica}, 3:127--128, 1983.

\bibitem{OV17}
N.~Olver and L.~A. V\'egh.
\newblock A simpler and faster strongly polynomial algorithm for generalized
  flow maximization.
\newblock In {\em Proceedings of STOC}, pages 100--111. ACM, 2017.

\bibitem{orlin93}
J.~B. Orlin.
\newblock A faster strongly polynomial minimum cost flow algorithm.
\newblock {\em Operations Research}, 41(2):338--350, 1993.

\bibitem{orlin}
J.~B. Orlin.
\newblock Improved algorithms for computing {Fisher}'s market clearing prices.
\newblock In {\em Proc.\ 42nd Symp.\ Theory of Computing (STOC)}, pages
  291--300, 2010.

\bibitem{Orlin13}
J.~B. Orlin.
\newblock Max flows in {$O(nm)$} time, or better.
\newblock In {\em Proceedings of STOC}, pages 765--774. ACM, 2013.

\bibitem{Shmyrev09}
V.~I. Shmyrev.
\newblock An algorithm for finding equilibrium in the linear exchange model
  with fixed budgets.
\newblock {\em Journal of Applied and Industrial Mathematics}, 3(4):505--518,
  2009.

\bibitem{Tardos85}
{\'E}.~Tardos.
\newblock A strongly polynomial minimum cost circulation algorithm.
\newblock {\em Combinatorica}, 5(3):247--255, 1985.

\bibitem{Tardos86}
{\'E}.~Tardos.
\newblock A strongly polynomial algorithm to solve combinatorial linear
  programs.
\newblock {\em Operations Research}, pages 250--256, 1986.

\bibitem{Varian74}
H.~Varian.
\newblock Equity, envy and efficiency.
\newblock {\em J. Econom.\ Theory}, 29(2):217--244, 1974.

\bibitem{Vavasis96}
S.~A. Vavasis and Y.~Ye.
\newblock A primal-dual interior point method whose running time depends only
  on the constraint matrix.
\newblock {\em Mathematical Programming}, 74(1):79--120, 1996.

\bibitem{Vazirani10}
V.~Vazirani.
\newblock Spending constraint utilities with applications to the adwords
  market.
\newblock {\em Mathematics of Operations Research}, 35(2):458--478, 2010.

\bibitem{vazirani12}
V.~V. Vazirani.
\newblock The notion of a rational convex program, and an algorithm for the
  {Arrow-Debreu Nash bargaining game}.
\newblock {\em Journal of the ACM (JACM)}, 59(2):7, 2012.

\bibitem{Vegh14}
L.~A. V{\'e}gh.
\newblock Concave generalized flows with applications to market equilibria.
\newblock {\em Mathematics of Operations Research}, 39(2):573--596, 2013.

\bibitem{Vegh16}
L.~A. V{\'e}gh.
\newblock A strongly polynomial algorithm for a class of minimum-cost flow
  problems with separable convex objectives.
\newblock {\em SIAM Journal on Computing}, 45(5):1729--1761, 2016.

\bibitem{Vegh17}
L.~A. V\'egh.
\newblock A strongly polynomial algorithm for generalized flow maximization.
\newblock {\em Mathematics of Operations Research}, 42(2):179--211, 2017.

\bibitem{walras}
L.~Walras.
\newblock El{\'e}ments d'{\'e}conomie politique pure, ou th{\'e}orie de la
  richesse sociale (in {French}), 1874.
\newblock English translation: {Elements of pure economics; or, the theory of
  social wealth}. {American Economic Association and the Royal Economic
  Society}, 1954.

\bibitem{Ye2005}
Y.~Ye.
\newblock A new complexity result on solving the {M}arkov decision problem.
\newblock {\em Mathematics of Operations Research}, 30(3):733--749, 2005.

\bibitem{Ye08}
Y.~Ye.
\newblock A path to the {Arrow--Debreu} competitive market equilibrium.
\newblock {\em Mathematical Programming}, 111(1-2):315--348, 2008.

\bibitem{Ye2011}
Y.~Ye.
\newblock The simplex and policy-iteration methods are strongly polynomial for
  the {M}arkov decision problem with a fixed discount rate.
\newblock {\em Mathematics of Operations Research}, 36(4):593--603, 2011.

\end{thebibliography}

\end{document}